\documentclass{article}

\usepackage{graphicx}
\usepackage{amsmath,amsthm,amssymb,amsfonts}
\usepackage{fullpage}
\usepackage{microtype}
\usepackage[utopia]{mathdesign}
\usepackage{stackrel}

\RequirePackage{epigraph}
\RequirePackage{pstricks}

\newcommand{\bigo}{O}
\newcommand{\eps}{\varepsilon}
\newcommand{\Syes}{S^\textrm{yes}}
\newcommand{\saccept}{s^\textrm{accept}}
\newcommand{\sinit}{s^0}
\newcommand{\DFA}{\textrm{DFA}}
\newcommand{\N}{\mathbb{N}}
\newcommand{\nobb}{L_\text{\rm no $bb$}}
\newcommand{\Lpal}{L_\textrm{pal}}
\newcommand{\Lcopy}{L_\textrm{copy}}
\renewcommand{\wp}{\mathcal{P}}
\newcommand{\blank}{\textvisiblespace}
\newcommand{\sort}{\mathrm{sort}}
\newcommand{\lcm}{\mathrm{lcm}}
\newcommand{\Z}{\mathbb{Z}}

\newtheorem{theorem}{Theorem}
\newtheorem{lemma}{Lemma}
\newtheorem{definition}{Definition}
\newtheorem{proposition}{Proposition}
\newtheorem{corollary}{Corollary}
\newtheorem{exercise}{Exercise}

\renewcommand{\emptyset}{\varnothing}

\newcommand{\loesung}[1]{}

\begin{document}

\title{Automata, languages, and grammars}
% \\ (with solutions to selected exercises)}
\author{Cristopher Moore}
\maketitle

\begin{abstract}
These lecture notes are intended as a supplement to Moore and Mertens' \emph{The Nature of Computation}, and are available to anyone who wants to use them.  
%Written solutions to some exercises are available upon request.  
Comments are welcome, and please let me know if you use these notes in a course.
\end{abstract}

\setlength{\epigraphwidth}{0.5 \textwidth}
\epigraph{\ldots nor did Pnin, as a teacher, ever presume to approach the lofty halls of modern scientific linguistics \ldots\ 
%that temple wherein earnest young people are taught not the language itself, but the method of teaching others to teach that method; 
which perhaps in some fabulous future may become instrumental in evolving esoteric dialects---Basic Basque and so forth---spoken only by certain elaborate machines.}{Vladimir Nabokov, \emph{Pnin}}

In the 1950s and 60s there was a great deal of interest in the power of simple machines, motivated partly by the nascent field of computational complexity and partly by formal models of human (and computer) languages.  These machines and the problems they solve give us simple playgrounds in which to explore issues like nondeterminism, memory, and information.  Moreover, they give us a set of complexity classes that, unlike P and NP, we can understand completely.  

In these lecture notes, we explore the most natural classes of automata, the languages they recognize, and the grammars they correspond to.  While this subfield of computer science lacks the mathematical depth we will encounter later on, it has its own charms, and provides a few surprises---like undecidable problems involving machines with a single stack or counter.

\section{Finite-State Automata}
\label{sec:dfa}

Here is a deterministic finite-state automaton, or DFA for short:

\begin{center}
\includegraphics[width=1.8in]{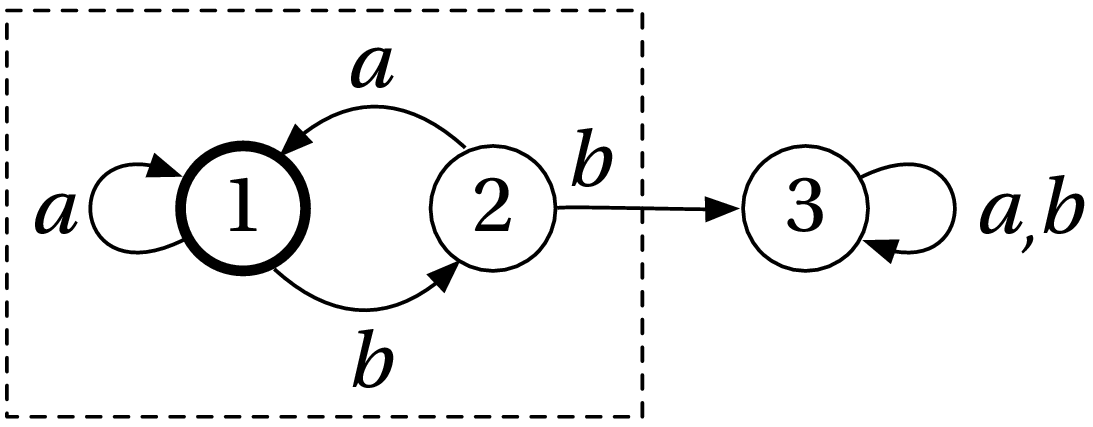}
\end{center}

\noindent
It has three states, $\{1,2,3\}$, and an input alphabet with two symbols, $\{a,b\}$.  It starts in the state $1$ (circled in bold) and reads a string of $a$s and $b$s, making transitions from one state to another along the arrows.  It is \emph{deterministic} because at each step, given its current state and the symbol it reads next, it has one and only one choice of state to move to.

This DFA is a tiny kind of computer.  It's job is to say ``yes'' or ``no'' to input strings---to ``accept'' or ``reject'' them.  It does this by reading the string from left to right, arriving in a certain final state.  If its final state is in the dashed rectangle, in state $1$ or $2$, it accepts the string.  Otherwise, it rejects it.  Thus any DFA answers a simple yes-or-no question---namely, whether the input string is in the set of strings that it accepts.  

The set of yes-or-no questions that can be answered by a DFA is a kind of baby complexity class.  It is far below the class P of problems that can be solved in polynomial time, and indeed the problems it contains are almost trivial.  But it is interesting to look at, partly because, unlike P, we understand exactly what problems this class contains.  To put this differently, while polynomial-time algorithms have enormous richness and variety, making it very hard to say exactly what they can and cannot do, we can say precisely what problems DFAs can solve.  

Formally, a DFA $M$ consists of a set of states $S$, an input alphabet $A$, an initial state $\sinit \in S$, a subset $\Syes \subseteq S$ of accepting states, and a transition function
\[
\delta : S \times A \to S \, ,
\]
that tells the DFA which state to move to next.  In other words, if it's in state $s$ and it reads the symbol $a$, it moves to a new state $\delta(s,a)$.  In our example, $S=\{1,2,3\}$, $A=\{a,b\}$, $\sinit=1$, $\Syes = \{1,2\}$, and the transition function is
\[
\begin{array}{cc}
\delta(1,a)=1 & \delta(1,b)= 2 \\
\delta(2,a)=1 & \delta(2,b)= 3 \\
\delta(3,a)=3 & \delta(3,b)= 3 
\end{array}
\]

Given a finite alphabet $A$, let's denote the set of all finite strings $A^*$.  For instance, if $A=\{a,b\}$ then 
\[
A^* = \{ \eps, a, b, aa, ab, ba, bb, \ldots\ \} \, .
\]
Here $\eps$ denotes the empty string---that is, the string of length zero.  Note that $A^*$ is an infinite set, but it only includes finite strings.  It's handy to generalize the transition function and let $\delta^*(s,w)$ denote the state we end up in if we start in a state $s$ and read a string $w \in A^*$ from left to right.  For instance, 
\[
\delta^*(s,abc) = \delta(\delta(\delta(s,a),b),c) \, . 
\]
We can define $\delta^*$ formally by induction, 
\[
\delta^*(s,w) = \begin{cases}
s & \mbox{if $w=\eps$} \\ 
\delta(\delta^*(s,u),a) & \mbox{if $w=ua$ for $u \in A^*$ and $a \in A$} \, . 
\end{cases}
\]
Here $ua$ denotes the string $u$ followed by the symbol $a$.  In other words, if $w$ is empty, do nothing.  Otherwise, read all but the last symbol of $w$ and then read the last symbol.  More generally, if $w=uv$, i.e., the concatenation of $u$ and $v$, an easy proof by induction gives
\[
\delta^*(s,w) = \delta^*(\delta^*(s,u),v) \, . 
\]
In particular, if the input string is $w$, the final state the DFA ends up in is $\delta^*(\sinit,w)$.  Thus the set of strings, or the \emph{language}, that $M$ accepts is
\[
L(M) = \{ w \in A^* \mid \delta^*(\sinit,w) \in \Syes \} \, . 
\]
If $L=L(M)$, we say that $M$ \emph{recognizes} $L$.  That is, $M$ answers the yes-or-no question if whether $w$ is in $L$, accepting $w$ if and only if $w \in L$.  For instance, our example automaton above recognizes the languages of strings with no two $b$s in a row, 
\[
\nobb = \{ \eps, a, b, aa, ab, ba, aaa, aab, aba, baa, bab, \ldots \} \, . 
\]

%\begin{exercise}[A bit of combinatorics]
%For each $\ell$, how many strings are there in $\nobb$ of length $\ell$?
%\end{exercise}

Some languages can be recognized by a DFA, and others can't.  Consider the following exercise:
\begin{exercise}
Show that for a given alphabet $A$, the set of possible DFAs is countably infinite (that is, it's the same as the number of natural numbers) while the set of all possible languages is uncountably infinite (that is, it's as large as the number of subsets of the natural numbers).
\end{exercise}

\loesung{
A DFA is finite object, and can be coded with a finite number of bits describing the number of states $|S|$, its transition function $\delta$, and the set of accepting states $\Syes$.  Thus it can be coded as an integer, and the number of DFAs is at most the number of integers.

On the other hand, the number of possible languages is uncountably infinite, even if the alphabet has just a single symbol $a$.  For any subset $T\subseteq \N$, we can define a language $L$ such that the word $a^i$ (i.e., $a$ repeated $i$ times) is in $L$ if and only if $i \in T$.  Thus the number of languages is at least as large as the number of subsets of $\N$.
}

\noindent
By Cantor's diagonalization proof, this shows that there are infinitely (transfinitely!) many more languages than there are DFAs.  In addition, languages like the set of strings of $0$s and $1$s that encode a prime number in binary seem too hard for a DFA to recognize.  Below we will see some ways to prove intuitions like these.  First, let's give the class of languages that can be recognized by a DFA a name.
\begin{definition}
If $A$ is a finite alphabet, a language $L \subseteq A^*$ is \emph{DFA-regular} if there is exists a DFA $M$ such that $L=L(M)$.  That is, $M$ accepts a string $w$ if and only if $w \in L$.  
\end{definition}
This definition lumps together DFAs of all different sizes.  In other words, a language is regular if it can be recognized by a DFA with one state, or two states, or three states, and so on.  However, the number of states has to be constant---it cannot depend on the length of the input word.

Our example above shows that the language $\nobb$ of strings of $a$s and $b$s with no $bb$ is DFA-regular.  What other languages are?
\begin{exercise}
\label{ex:dfa}
Show that the the following languages are DFA-regular.
\begin{enumerate}
\item The set of strings in $\{a,b\}^*$ with an even number of $b$'s.
\item The set of strings in $\{a,b,c\}^*$ where there is no $c$ anywhere to the left of an $a$.
\item The set of strings in $\{0,1\}^*$ that encode, in binary, an integer $w$ that is a multiple of\, $3$.  Interpret the empty string $\eps$ as the number zero.  
%Hint: what exactly does the automaton need to remember about the bits of $w$ that it has seen so far?
\end{enumerate}
In each case, try to find the minimal DFA $M$, i.e., the one with the smallest number of states, such that $L=L(M)$.  Offer some intuition about why you believe your DFA is minimal.
\end{exercise}

\loesung{
Here are DFAs for each of these languages.  The initial state is in bold, and the accepting states are outlined in a dashed box.
\begin{figure}[htb]
\begin{center}
\includegraphics[width=0.7\textwidth]{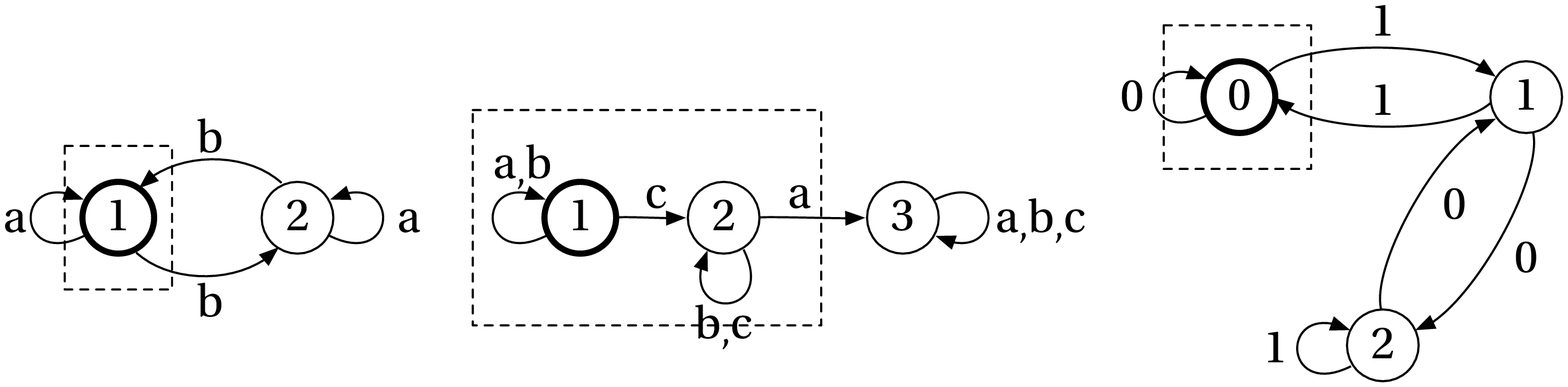}
\end{center}
\label{fig:3dfas}
\end{figure}

\noindent
The first two are self-explanatory.  The third one is more interesting.  If $x$ is the integer read so far, the states $\{0,1,2\}$ correspond to the value of $x \bmod 3$.  Reading a $0$ or $1$ sends $x$ to $2x$ or $2x+1$ respectively.  Do you see how this would generalize if the alphabet were $\{0,\ldots,9\}$, say, and the word had to correspond to a multiple of $7$ written in base $10$?
}

\noindent
As you do the previous exercise, a theme should be emerging.  The question about each of these languages is this: as you read a string $w$ from left to right, what do you need to keep track of at each point in order to be able to tell, at the end, whether $w \in L$ or not?  

\begin{exercise}
Show that any finite language is DFA-regular.
\end{exercise}

\loesung{
If $L$ is finite, then its longest word has some finite length $\ell$.  We can easily write down a DFA with $1+|A|+|A|^2+\cdots+|A|^\ell = O(|A|^\ell)$ states corresponding to the words $w$ of length at most $\ell$: its transition diagram is a tree, where the root corresponds to the empty word, the first level corresponds to the $|A|$ words of length $1$, and so on.  We define each state as accepting or rejecting according to whether or not $w \in L$.  We can add an additional rejecting state to handle the case where the input has length greater than $\ell$.
}

Now that we have some examples, let's prove some general properties of regular languages.

\begin{exercise}[Closure under complement]
Prove that if $L$ is DFA-regular, then $\overline{L}$ is DFA-regular.
\end{exercise}

\loesung{
If $L$ is recognized by a DFA $M$, then replacing $\Syes$ with $\overline{\Syes}$ switches ``yes'' and ``no'' outputs, and yields a DFA $M'$ that recognizes $\overline{L}$.
}

\noindent
This is a simple example of a \emph{closure property}---a property saying that the set of DFA-regular languages is closed under certain operations..  Here is a more interesting one:

\begin{proposition}[Closure under intersection]
\label{prop:intersection}
If $L_1$ and $L_2$ are DFA-regular, then $L_1 \cap L_2$ is DFA-regular.
\end{proposition}

\begin{proof}
The intuition behind the proof is to run both automata at the same time, and accept if they both accept.  Let $M_1$ and $M_2$ denote the automata that recognize $L_1$ and $L_2$ respectively.  They have sets of states $S_1$ and $S_2$, initial states $\sinit_1$ and $\sinit_2$, and so on.  Define a new automaton $M$ as follows:
\[
S = S_1 \times S_2 , \quad \sinit = (\sinit_1, \sinit_2) , \quad \Syes = \Syes_1 \times \Syes_2 , \quad
\delta\big( (s_1,s_2), a \big) = \big( \delta_1(s_1,a) , \delta_2(s_2,a) \big) \, . 
\]
Thus $M$ runs both two automata in parallel, updating both of them at once, and accepts $w$ if they both end in an accepting state.  To complete the proof in gory detail, by induction on the length of $w$ we have
\[
\delta^*(\sinit,w) = \big( \delta^*_1(\sinit_1,w), \delta^*_2(\sinit_2,w) \big) \, ,
\]
Since $\delta^*(\sinit,w) \in \Syes$ if and only if $\delta^*_1(\sinit_1,w) \in \Syes_1$ and $\delta^*_2(\sinit_2,w) \in \Syes_2$, in which case $w \in L_1$ and $w \in L_2$, $M$ recognizes $L_1 \cap L_2$.
\end{proof}

Note that this proof creates a combined automaton with $|S| = |S_1| |S_2|$ states.  Do you we think we can do better?  Certainly we can in some cases, such as if $L_2$ is the empty set.  But do you think we can do better in general?  Or do you think there is an infinite family of pairs of automata $(M_1,M_2)$ of increasing size such that, for every pair in this family, the smallest automaton that recognizes $L(M_1) \cap L(M_2)$ has $|S_1| |S_2|$ states?  We will resolve this question later on---but for now, see if you can come up with such a family, and an intuition for why $|S_1| |S_2|$ states are necessary.

Now that we have closure under complement and intersection, de Morgan's law 
\[
L_1 \cup L_2 = \overline{\overline{L_1} \cap \overline{L_2}} 
\]
tells us that the union of two DFA-regular languages is DFA-regular.  Thus we have closure under union as well as intersection and complement.  By induction, any Boolean combination of DFA-regular languages is DFA-regular.

\begin{exercise}
Pretend you don't know de Morgan's law, and give a direct proof in the style of Proposition~\ref{prop:intersection} that the set of DFA-regular languages is closed under union.
\end{exercise}

\loesung{Let $M_1$ and $M_2$ be DFAs that recognize $L_1$ and $L_2$ respectively.  We will construct a DFA $M$ that recognizes $L_1 \cup L_2$.  As in Proposition~\ref{prop:intersection}, $M$ runs $M_1$ and $M_2$ in parallel, with state space $S = S_1 \times S_2$, initial state $\sinit = (\sinit_1, \sinit_2)$, and transitions $\delta\big( (s_1,s_2), a \big) = \big( \delta_1(s_1,a) , \delta_2(s_2,a) \big)$.  The only difference is that now we define its accepting states as
\[
S = \left\{ (s_1,s_2) \mid \mbox{$s_1 \in \Syes_1$ or $s_2 \in \Syes_2$} \right\} = \left( \Syes_1 \times S_2 \right) \cup \left( S_1 \times \Syes_2 \right) \, . 
\]
Then $M$ accepts if either $M_1$ or $M_2$ accept.
}

\begin{exercise}
%Suppose that $L_1$ and $L_2$ are regular languages, and that they are recognized by DFAs with $n_1$ states and $n_2$ states respectively.  Say that a word $w$ \emph{distinguishes} $L_1$ and $L_2$ if $w \in L_1$ but $w \notin L_2$, or vice versa.  Show that unless $L_1$ and $L_2$ are identical, there is a word $w$ of length at most $n_1 n_2$ that distinguishes them.
Consider two DFAs, $M_1$ and $M_2$, with $n_1$ and $n_2$ states respectively.  Show that if $L(M_1) \ne L(M_2)$, there is at least one word of length less than $n_1 n_2$ that $M_1$ accepts and $M_2$ rejects or vice versa.  Contrapositively, if all words of length less than $n_1 n_2$ are given the same response by both DFAs, they recognize the same language.
\end{exercise}

\loesung{
Once again, we combine $M_1$ and $M_2$ into a larger DFA $M$, but this time one that recognizes the words on which $M_1$ and $M_2$ disagree.  As before, set $S = S_1 \times S_2$, $\sinit = (\sinit_1, \sinit_2)$, and $\delta\big( (s_1,s_2), a \big) = \big( \delta_1(s_1,a) , \delta_2(s_2,a) \big)$.  Now define the accepting states as
\begin{align*}
S 
&= \left\{ (s_1,s_2) \mid \mbox{($s_1 \in \Syes_1$ and $s_2 \notin \Syes_2$) or ($s_1 \notin \Syes_1$ and $s_2 \in \Syes_2$)} \right\} \\
&= \left( \Syes_1 \times \overline{\Syes_2} \right) \cup \left( \overline{\Syes_1} \times \Syes_2 \right) \, . 
\end{align*}
This DFA has $n_1 n_2$ states, and accepts $(L_1 \cap \overline{L_2}) \cup (\overline{L_1} \cap L_2)$, the set of words that are in one language but not the other.  Finally, note that if a DFA with $n$ states accepts any words, then it must accept at least one word of length less than $n$, since there is a path from $\sinit$ to some $s \in \Syes$ consisting of $n-1$ or fewer steps.
}

These closure properties correspond to ways to modify a DFA, or combine two DFAs to create a new one.  We can switch the accepting and rejecting states, or combine two DFAs so that they run in parallel.  What else can we do?  Another way to combine languages is concatenation.  Given two languages $L_1$ and $L_2$, define the language
\[
L_1 L_2 = \{ w_1 w_2 \mid w_1 \in L_1, w_2 \in L_2 \} \, . 
\]
In other words, strings in $L_1 L_2$ consist of a string in $L_1$ followed by a string in $L_2$.  Is the set of DFA-regular languages closed under concatenation?

In order to recognize $L_1 L_2$, we would like to run $M_1$ until it accepts $w_1$, and then start running $M_2$ on $w_2$.  But there's a problem---we might not know where $w_1$ ends and $w_2$ begins.  For instance, let $L_1$ be $\nobb$ from our example above, and let $L_2$ be the set of strings from the first part of Exercise~\ref{ex:dfa}, where there are an even number of $b$s.  Then the word $babaabbab$ is in $L_1 L_2$.  But is it $babaab+bab$, or $ba+baabbab$, or $b+abaabbab$?  In each case we could jump from an accepting state of $M_1$ to the initial state of $M_2$.  But if we make this jump in the wrong place, we could end up rejecting instead of accepting: for instance, if we try $babaa+bbab$.  

If only we were allowed to guess where to jump, and somehow always be sure that if we'll jump at the right place, if there is one\ldots

\section{Nondeterministic Finite-State Automata}
\label{sec:nfa}

Being deterministic, a DFA has no choice about what state to move to next.  In the diagram of the automaton for $\nobb$ above, each state has exactly one arrow pointing out of it labeled with each symbol in the alphabet.  What happens if we give the automaton a choice, giving some states multiple arrows labeled with the same symbol, and letting it choose which way to go?  

This corresponds to letting the transition function $\delta(s,a)$ be multi-valued, so that it returns a set of states rather than a single state.  Formally, it is a function into the power set $\wp(S)$, i.e., the set of all subsets of $S$:
\[
\delta: S \times A \to \wp(S) \, . 
\]
We call such a thing a \emph{nondeterministic finite-state automaton}, or NFA for short.

Like a DFA, an NFA starts in some initial state $\sinit$ and makes a series of transitions as it reads a string $w$.  But now the set of possible computations branches out, letting the automaton follow many possible paths.  Some of these end in the accepting state $\Syes$, and others don't.  Under what circumstances should we say that $M$ accepts $w$?  How should we define the language $L(M)$ that it recognizes?  There are several ways we might do this, but the one we use is this: $M$ accepts $w$ if and only if \emph{there exists a computation path} ending in an accepting state.  Conversely, $M$ rejects $w$ if and only if \emph{all computation paths} end in a rejecting state.

Note that this notion of ``nondeterminism'' has nothing to do with probability.  It could be that only one out of exponentially many possible computation paths accepts, so that if the NFA flips a coin each time it has a choice, the probability that it finds this path is exponentially small.  We judge acceptance not on the probability of an accepting path, but simply on its existence.  This definition may seem artificial, but it is analogous to the definition of NP (see Chapter 4) where the answer to the input is ``yes'' if there exists a solution, or ``witness,'' that a deterministic polynomial-time algorithm can check.  

Perhaps it's time for an example.  Let the input alphabet be $A=\{a,b\}$, and define an NFA as follows:
\begin{center}
\includegraphics[width=2.3in]{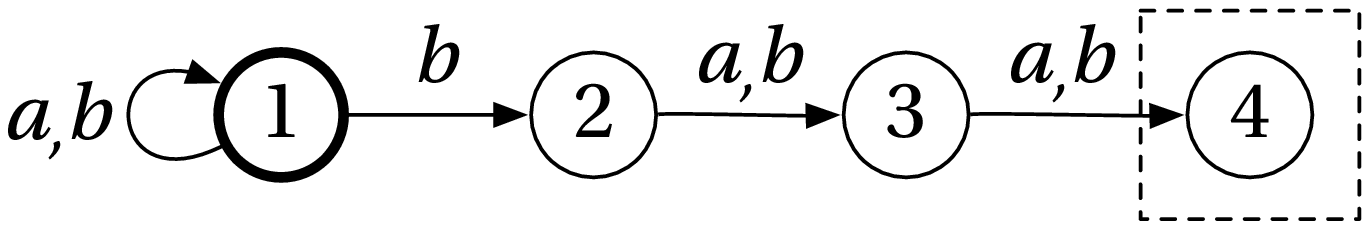}
\end{center}
Let's call this automaton $M_3$.  It starts in state $1$, and can stay in state $1$ as long as it likes.  But when it reads the symbol $b$, it can move to state $2$ if it prefers.  After that, it moves inexorably to state $3$ and state $4$ regardless of the input symbol, accepting only if it ends in state $4$.  There are no allowed transitions from state $4$.  In other words, $\delta(4,a)=\delta(4,b)=\emptyset$.

We call this automaton $M_3$ because it accepts the set of strings whose 3rd-to-last symbol is a $b$.  But it has to get its timing right, and move from state $1$ to state $2$ when it sees that $b$.  Otherwise, it misses the boat.  How hard do you think it is to recognize this set deterministically?  Take a moment to design a DFA that recognizes $L(M_3)$.  How many states must it have?  As in Exercise~\ref{ex:dfa}, how much information does it need to keep track of as it reads $w$ left-to-right, so that whenever $w$ steps, it ``knows'' whether the 3rd-to-last symbol was a $b$ or not?

As we did for DFAs, let's define the class of languages that can be recognized by some NFA:

\begin{definition}
A language $L \subseteq A^*$ is \emph{NFA-regular} if there is exists a NFA $M$ such that $L=L(M)$.  
\end{definition}

\noindent
A DFA is a special case of an NFA, where $\delta(s,a)$ always contains a single state.  Thus any language which is DFA-regular is automatically also NFA-regular.  

On the other hand, with their miraculous ability to guess the right path, it seems as if NFAs could be much more powerful than DFAs.  Indeed, where P and NP are concerned, we believe that nondeterminism makes an enormous difference.  But down here in the world of finite-state automata, it doesn't, as the following theorem shows.

\begin{theorem}
\label{thm:nfa-dfa}
Any NFA can be simulated by a DFA that accepts the same language.  Therefore, a language $L$ is NFA-regular if and only if it is DFA-regular.
\end{theorem}

\begin{proof}
The idea is to keep track of all possible states that an NFA can be in at each step.  While the NFA's transitions are nondeterministic, this set changes in a deterministic way.  Namely, after reading a symbol $a$, a state $s$ is possible if $s \in \delta(t,a)$ for some state $t$ that was possible on the previous step.

This lets us simulate an NFA $M$ with a DFA $M_\DFA$ as follows.  If $M$'s set of states is $S$, then $M_\DFA$'s set of states $S_\DFA$ is the power set $\wp(S) = \{ T \subseteq S \}$, its initial state $\sinit_\DFA$ is the set $\{\sinit\}$, and its transition function is 
\[
\delta_\DFA(T,a) = \bigcup_{t \in T} \delta_M(t,a) \, . 
\]
Then for any string $w$, $\delta^*_\DFA(\sinit_\DFA,w)$ is the set of states that $M$ could be in after reading $w$.  Finally, $M$ accepts if it could end in at least one accepting state, so we define the accepting set of $M_\DFA$ as 
\[
\Syes_\DFA = \{ T \cap \Syes \ne \emptyset \} \, . 
\]
Then $L(M_\DFA)=L(M)$.
\end{proof}

While this theorem shows that any NFA can be simulated by a DFA, the DFA is much larger.  If $M$ has $|S|=n$ states, then $M_\DFA$ has $|\wp(S)|=2^n$ states.  Since $2^n$ is finite for any finite $n$, and since the definition of DFA-regular lets $M_\DFA$ have any finite size, this shouldn't necessarily bother us---the size of $M_\DFA$ is ``just a constant.''  But as for our question about $L_1 \cap L_2$ above, do you think this is the best we can do?

Now that we know that DFA-regular and NFA-regular languages are the same, let's use the word ``regular'' for both of them.  Having two definitions for the same thing is helpful.  If we want to prove something about regular languages, we are free to use whichever definition---that is, whichever type of automaton---makes that proof easier.

For instance, let's revisit the fact that the union of two regular languages is regular.  In Section~\ref{sec:dfa}, we proved this by creating a DFA that runs two DFAs $M_1$ and $M_2$ in parallel.  But an NFA can do something like this:
\begin{center}
\includegraphics[width=0.9in]{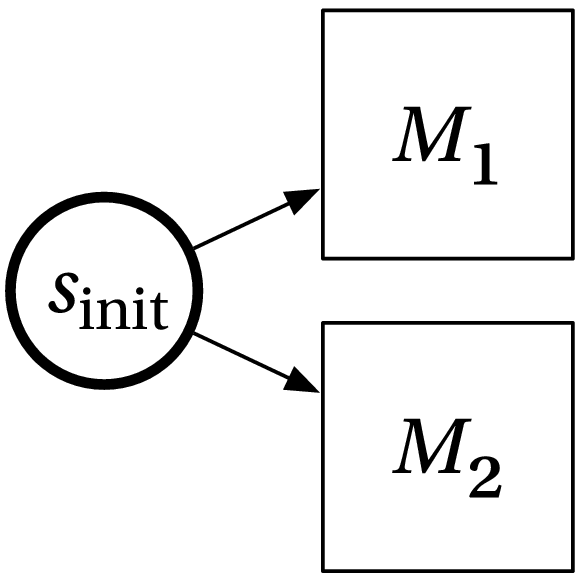}
\end{center}
This NFA guesses which of $M_1$ or $M_2$ it should run, rather than running both of them at once, and recognizes $L_1 \cup L_2$ with about $|S_1| + |S_2|$ states instead of $|S_1| |S_2|$ states.  To make sense, the edges out of $\sinit$ have to be \emph{$\eps$-transitions}.  That is, the NFA has to be able to jump from $\sinit$ to $\sinit_1$ or $\sinit_2$ without reading any symbols at all:
\[
\delta^*(\sinit,\eps) = \{ \sinit_1, \sinit_2 \} \, . 
\]
Allowing this kind of transition in our diagram doesn't change the definition of NFA at all (exercise for the reader).  This also makes it easy to prove closure under concatenation, which we didn't see how to do with DFAs:

\begin{proposition}
If $L_1$ and $L_2$ are regular, then $L_1 L_2$ is regular.
\end{proposition}

\begin{proof}
Start with NFAs $M_1$ and $M_2$ that recognize $L_1$ and $L_2$ respectively.  We assume that $S_1$ and $S_2$ are disjoint.  Define a new NFA $M$ with $S=S_1 \cup S_2$, $\sinit = \sinit_1$, $\Syes = \Syes_2$, and an $\eps$-transition from each $s_1 \in \Syes_1$ to $\sinit_2$.  Then $M$ recognizes $L_1 L_2$.
\end{proof}

Another important operator on languages is the \emph{Kleene star}.  Given a language $L$, we define $L^*$ as the concatenation of $t$ strings from $L$ for any integer $t \ge 0$:
\begin{align*}
L^*
&= {\eps} \cup L \cup LL \cup LLL \cup \cdots \\
&= \left\{ w_1 w_2 \cdots w_t \mid \mbox{$t \ge 0, w_i \in L$ for all $1 \le i \le t$} \right\} \, . 
\end{align*}
This includes our earlier notation $A^*$ for the set of all finite sequences of symbols in $A$.  Note that $t=0$ is allowed, so $L^*$ includes the empty word $\eps$.  Note also that $L^t$ doesn't mean repeating the same string $t$ times---the $w_i$ are allowed to be different.  The following exercise shows that the class of regular languages is closed under $*$:

\begin{exercise}
Show that if $L$ is regular then $L^*$ is regular.  Why does it not suffice to use the fact that the regular languages are closed under concatenation and union?
\end{exercise}

\loesung{Given an NFA $M$ for $L$, create an NFA $M'$ for $L^*$ simply by adding $\eps$-transitions from each accepting $s \in \Syes$ state back to the initial state $\sinit$.  This lets us concatenate accepting paths, so if $w_1, \ldots, w_t \in L$ then there is an accepting path for $w_1 w_2 \cdots w_t \in L^*$.  Conversely, any accepting path $p$ in $M'$ consists of a series of accepting paths in $M$, namely the segments of $p$ that begin at $\sinit$ and end in $\Syes$.  In order to make sure that $\eps$ is accepted, we also include an $\eps$-transition from $\sinit$ to an additional accepting state.

If we try to do this using closure under concatenation and union, we can certainly make an NFA that recognizes $L \cup LL$, another that recognizes $L \cup LL \cup LLL$, and so on.  But these automata get bigger and bigger, and we want a single NFA that recognizes words in $L^t$ for any $t \ge 0$.  Hence we need some kind of loop in the NFA's state space.
}

Here is another fact that is easier to prove using NFAs than DFAs:

\begin{exercise}
\label{ex:reverse}
Given a string $w$, let $w^R$ denote $w$ written in reverse.  Given a language $L$, let $L^R = \{ w^R \mid w \in L \}$.  Prove that $L$ is regular if and only if $L^R$ is regular.  Why is this harder to prove with DFAs?
\end{exercise}

\loesung{
Given an NFA $M$ for $L$, our instinct is to create an NFA $M'$ for $L^R$ by reversing all the arrows, and making the new $\Syes$ the old $\sinit$ and vice versa.  This almost works, but since NFAs are defined with one initial state but multiple accepting states, we create a new initial state ${\sinit}'$ and add $\eps$-transitions from it to each $s \in \Syes$.  Finally we set ${\Syes}' = \{\sinit\}$.

If we try to do this with a DFA, we will usually create some nondeterminism, giving an NFA instead.  If two arrows in a DFA arrive at the same state when reading the same symbol, reversing them creates a branch point where we can go in two different directions.
}

On the other hand, if you knew about NFAs but not about DFAs, it would be tricky to prove that the complement of a regular language is regular.  The definition of nondeterministic acceptance is asymmetric: a string $w$ is in $\overline{L(M)}$ if \emph{every} computation path leads to a state in $\overline{\Syes}$.  Logically speaking, the negation of a ``there exists'' statement is a ``for all'' statement, creating a different kind of nondeterminism.  Let's show that defining acceptance this way again keeps the class of regular languages the same:

\begin{exercise}
A \emph{for-all NFA} is one such that $L(M)$ is the set of strings where \emph{every} computation path ends in an accepting state.  Show how to simulate an for-all NFA with a DFA, and thus prove that a language is recognized by some for-all NFA if and only if it is regular.
\end{exercise}

\loesung{
Simulate an NFA $M$ with a DFA $M_\DFA$ as in Theorem~\ref{thm:nfa-dfa}, but now define the set of accepting states as
\[
\Syes_\DFA = \{ T \subseteq \Syes \} \, . 
\]
This forces every state that $M$ could be in at the end of computation to be an accepting state.
}

If that was too easy, try this one:

\begin{exercise} 
A \emph{parity finite-state automaton}, or PFA for short, is like an NFA except that it accepts a string $w$ if and only if the number of accepting paths induced by reading $w$ is odd.  Show how to simulate a PFA with a DFA, and thus prove that a language is recognized by a PFA if and only if it is regular.  Hint: this is a little trickier than our previous simulations, but the number of states of the DFA is the same.
\end{exercise}

\loesung{
Given a PFA with states $S$, we define a DFA with $S_\DFA = \{0,1\}^S$, the set of functions $f: S \to \{0,1\}$.  At each step, $f(s)$ will be the number of paths, mod $2$, from $\sinit$ to $s$.  The transition function updates this function by summing over all ways to extend the path from the previous state,
\[
\delta_\DFA(f,a) = f'
\quad \text{where} \quad
f'(t) = \sum_{s: t \in \delta(s,a)} f(s) \bmod 2 \, . 
\]
The set of accepting states is 
\[
\Syes_\DFA = \left\{ f \; \left\vert \; \sum_{s \in \Syes} f(s) \bmod 2 = 1 \right. \right\} \, . 
\]
}

Here is an interesting closure property:

\begin{exercise}
Given finite words $u$ and $v$, say that a word $w$ is an \emph{interweave} of $u$ and $v$ if I can get $w$ by peeling off symbols of $u$ and $v$, taking the next symbol of $u$ or the next symbol of $v$ at each step, until both are empty.  (Note that $w$ must have length $|w|=|u|+|v|$.)  For instance, if $u={\tt cat}$ and $v={\tt tapir}$, then one interleave of $u$ and $v$ is $w={\tt ctaapitr}$.  Note that, in this case, we don't know which ${\tt a}$ in $w$ came from $u$ and which came from $v$.

Now given two languages $L_1$ and $L_2$, let $L_1 \wr L_2$ be the set of all interweaves $w$ of $u$ and $v$, for all $u \in L_1$ and $v \in L_2$.  Prove that if $L_1$ and $L_2$ are regular, then so is $L_1 \wr L_2$.
\end{exercise}

\loesung{
Let $M_1$ and $M_2$ be NFAs that recognize $L_1$ and $L_2$.  We define an NFA that recognizes $L_1 \wr L_2$.  It is similar to the DFA we defined in Proposition~\ref{prop:intersection} for $L_1 \cap L_2$.  It has states $S = S_1 \times S_2$, initial state $\sinit = (\sinit_1, \sinit_2)$, and accepting states $\Syes = \Syes_1 \times \Syes_2$.  But now, rather than making a transition in both machines at once, we (nondeterministically) make a transition in one machine or the other:
\begin{align*}
\delta\big( (s_1,s_2), a \big) 
&= \left\{ (s'_1,s'_2) \mid 
\left(s'_1 \in \delta_1(s_1,a) \mbox{ and } s'_2=s_2 \right)
\mbox{ or } 
\left(s'_1=s_1 \mbox{ and } s'_2 \in \delta_2(s_2,a) \right)
\right\} \\
&= \left( \delta_1(s_1,a) \times \{s_2\} \right) \;\bigcup\; \left( \{s_1\} \times \delta_2(s_2,a) \right)\, . 
\end{align*}
This automaton treats each symbol as belonging to a word in $L_1$ or a word in $L_2$, moving either $M_1$ or $M_2$ forward at each step, and accepts if both accept.
}

Finally, the following exercise is a classic.  If you are familiar with modern culture, consider the plight of River Song and the Doctor.

\begin{exercise}
Given a language $L$, let $L_{1/2}$ denote the set of words that can appear as first halves of words in $L$:
\[
L_{1/2} = \left\{ x \mid \exists y : |x|=|y| \mbox{ and } xy \in L \right\} \, ,
\]
where $|w|$ denotes the length of a word $w$.  Prove that if $L$ is regular, then $L_{1/2}$ is regular.  Generalize this to $L_{1/3}$, the set of words that can appear as middle thirds of words in $L$:
\[
L_{1/3} = \left\{ y \mid \exists x, z : |x|=|y|=|z| \mbox{ and } xyz \in L \right\} \, . 
\]
\end{exercise}

\loesung{
(Sketch) If $L$ is regular, it is recognized by some DFA $M$ with transition function $\delta$, initial state $\sinit$, and accepting states $\Syes$.  We will define an NFA $M_{1/2}$ that recognizes $L_{1/2}$.  It keeps track of a pair of states $(s,t)$ of $M$: initially, $s=\sinit$ and $t \in \Syes$.  At each step, when it reads a symbol $a$, $M_{1/2}$ updates $s$ to $\delta(s,a)$ (moving forward in time), and it updates $t$ to any $t'$ such that $t=\delta(t',a')$ for some $a'$ (moving backwards in time).  Then $M_{1/2}$ accepts if these forwards and backwards paths meet in the middle, i.e., if $s=t$.

To recognized $L_{1/3}$, take one step forward and two steps back.
}

\section{Equivalent States and Minimal DFAs}
\label{sec:minimal}

%In the previous sections, we asked several times 
The key to bounding the power of a DFA is to think about what kind of information it can gather, and retain, about the input string---specifically, how much it needs to remember about the part of the string it has seen so far.  Most importantly, a DFA has no memory beyond its current state.  It has no additional data structure in which to store or remember information, nor is it allowed to return to earlier symbols in the string.  

In this section, we will formalize this idea, and use it to derive lower bounds on the number of states a DFA needs to recognize a given language.  In particular, we will show that some of the constructions of DFAs in the previous sections, for the intersection of two regular languages or to deterministically simulate an NFA, are optimal.  

\begin{definition}
Given a language $L \subseteq A^*$, we say that a pair of strings $u, v \in A^*$ are \emph{equivalent with respect to $L$}, and write $u \sim_L v$, if for all $w \in A^*$ we have
\[
uw \in L \quad \text{if and only if} \quad vw \in L \, . 
\]
\end{definition}

\noindent
Note that this definition doesn't say that $u$ and $v$ are in $L$, although the case $w=\eps$ shows that $u \in L$ if and only if $v \in L$.  It says that $u$ and $v$ can end up in $L$ by being concatenated with the same set of strings $w$.  If you think of $u$ and $v$ as prefixes, forming the first part of a string, then they can be followed by the same set of suffixes $w$.

It's easy to see that $\sim_L$ is an equivalence relation: that is, it is reflexive, transitive, and symmetric.  This means that for each string $u$ we can consider its \emph{equivalence class}, the set of strings equivalent to it.  We denote this as
\[
[u] = \{ v \in A^* \mid u \sim_L v \} \, ,
\]
in which case $[u]=[v]$ if and only if $u \sim_L v$.  Thus $\sim_L$ carves the set of all strings $A^*$ up into equivalence classes.  

As we read a string from left to right, we can lump equivalent strings together in our memory.  We just have to remember the equivalence class of what we have seen so far, since every string in that class behaves the same way when we add more symbols to it.  For instance, $\nobb$ has three equivalence classes:
\begin{enumerate}
\item $[\eps]=[a]=[ba]=[abaaba]$, the set of strings with no $bb$ that do not end in $b$.
\item $[b]=[ab]=[baabab]$, the set of strings with no $bb$ that end in $b$.
\item $[bb]=[aababbaba]$, the set of strings with $bb$.
\end{enumerate} 
Do you see how these correspond to the three states of the DFA?

\begin{figure}
\begin{center}
\includegraphics[width=1.9in]{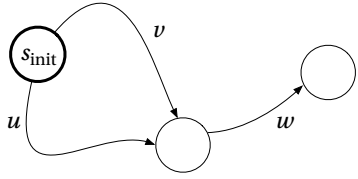}
\end{center}
\caption{If reading $u$ or $v$ puts $M$ in the same state, then so will reading $uw$ or $vw$, causing $M$ to accept both or reject both.}
\label{fig:uvw}
\end{figure}

On the other hand, if two strings $u, v$ are not equivalent, we had better be able to distinguish them in our memory.  Given a DFA $M$, let's define another equivalence relation, 
\[
u \sim_M v \;\Leftrightarrow\; \delta^*(\sinit,u) = \delta^*(\sinit,v)
\]
where $\sinit$ and $\delta$ are $M$'s initial state and transition function.  Two strings $u, v$ are equivalent with respect to $\sim_M$ if reading them puts $M$ in the same state as in Figure~\ref{fig:uvw}.  But once we're in that state, if we read a further word $w$, we will either accept both $uw$ and $vw$ or reject both.  
%But if $u \not\sim_L v$, there exists a string $w$ such that $uw \in L$ but $vw \notin L$ or vice versa, and we'll make a mistake in one of these cases.  
We prove this formally in the following proposition.

\begin{proposition}
\label{prop:distinguish}
Suppose $L$ is a regular language, and let $M$ be a DFA that recognizes $L$.  Then for any strings $u, v \in A^*$, 
\[
%\text{if $\delta^*(\sinit,u) = \delta^*(\sinit,v)$, 
\text{if $u \sim_M v$ 
then $u \sim_L v$} \, . 
\]
\end{proposition}

\begin{proof}
Suppose $\delta^*(\sinit,u) = \delta^*(\sinit,v)$.  Then for any string $w$, we have 
\[
\delta^*(\sinit,uw) 
= \delta^*\big( \delta^*(\sinit,u), w \big) 
= \delta^*\big( \delta^*(\sinit,v), w \big) 
= \delta^*(\sinit,vw) \, . 
\]
Thus $uw$ and $vw$ lead $M$ to the same final state.  This state is either in $\Syes$ or not, so $M$ either accepts both $uw$ and $vw$ or rejects them both.  If $M$ recognizes $L$, this means that $uw \in L$ if and only if $vw \in L$, so $u \sim_L v$.
\end{proof}

\noindent
Contrapositively, if $u$ and $v$ are not equivalent with respect to $\sim_L$, they cannot be equivalent with respect to $\sim_M$:
\[
\text{if $u \not\sim_L v$, then 
$u \not\sim_M v$} \, .
%$\delta^*(\sinit,u) \ne \delta^*(\sinit,v)$} \, .
\]
Thus each equivalence class requires a different state.  This gives us a lower bound on the number of states that $M$ needs to recognize $L$:
\begin{corollary}
\label{cor:myhill}
Let $L$ be a regular language.  If\, $\sim_L$ has $k$ equivalence classes, then any DFA that recognizes $L$ must have at least $k$ states.
\end{corollary}

We can say more than this.  The number of states of the minimal DFA that recognizes $L$ is exactly equal to the number of equivalence classes.  More to the point, the states and equivalence classes are in one-to-one correspondence.  To see this, first do the following exercise:

\begin{exercise}
\label{ex:sim}
Show that if $u \sim_L v$, then $ua \sim_L va$ for any $a \in A$.
\end{exercise}

\loesung{
If $u \sim_L v$, then $uw \in L$ if and only if $vw \in L$, for all $w \in A^*$.  This includes the case where $w$ begins with $a$, i.e., where $w=ay$ for some $y \in A^*$.  Thus $uay \in L$ if and only if $vay \in L$ for all $y \in A^*$, and $ua \sim_L va$.
}

\noindent
Thus for any equivalence class $[u]$ and any symbol $a$, we can unambiguously define an equivalence class $[ua]$.  That is, there's no danger that reading a symbol $a$ sends two strings in $[u]$ to two different equivalence classes.  This gives us our transition function, as described in the following theorem.

\begin{theorem}[Myhill-Nerode]
\label{thm:myhill-nerode}
Let $L$ be a regular language.  Then the minimal DFA for $L$, which we denote $M_{\min}$, can be described as follows.  It has one state for each equivalence class $[u]$.  Its initial state is $[\eps]$, its transition function is 
\[
\delta([u],a) = [ua] \, ,
\]
and its accepting set is 
\[
\Syes = \left\{ [u] \mid u \in L \right\} \, . 
\]
\end{theorem}

\noindent
This theorem is almost tautological at this point, but let's go through a formal proof to keep ourselves sharp.

\begin{proof}
We will show by induction on the length of $w$ that $M_{\min}$ keeps track of $w$'s equivalence class.  The base case is clear, since we start in $[\eps]$, the equivalence class of the empty word.  The inductive step follows from 
\[
\delta^*([\eps],wa) = \delta\big( \delta^*([\eps],w), a \big) = \delta([w],a) = [wa] \, . 
\]
Thus we stay in the correct equivalence class each time we read a new symbol.  This shows that, for all strings $w$, 
\[
\delta^*([\eps],w) = [w] \, . 
\]
Finally, $M_{\min}$ accepts $w$ if and only if $[w] \in \Syes$, which by the definition of $\Syes$ means that $w \in L$.  

Thus $M_{\min}$ recognizes $L$, and Corollary~\ref{cor:myhill} shows that any $M$ that recognizes $L$ has at least as many states as $M_{\min}$.  Therefore, $M_{\min}$ is the smallest possible DFA that recognizes $L$.
\end{proof}

Theorem~\ref{thm:myhill-nerode} also shows that the minimal DFA is \emph{unique up to isomorphism}.  That is, any two DFAs that both recognize $L$, and both have a number of states equal to the number of equivalence classes, have the same structure: there is a one-to-one mapping between their states that preserves the transition function, since both of them correspond exactly to the equivalence classes.

What can we say about non-minimal DFAs?  Suppose that $M$ recognizes $L$.  Proposition~\ref{prop:distinguish} shows that $\sim_M$ can't be a coarser equivalence than $\sim_L$.  That is, $\sim_M$ can't lump together two strings that aren't equivalent with respect to $\sim_L$.  But $\sim_M$ could be finer than $\sim_L$, distinguishing pairs of words that it doesn't have to in order to recognize $L$.  In general, the equivalence classes of $\sim_M$ are pieces of the equivalence classes of $\sim_L$, as shown in 
%If $u \sim_M v$ then $u \sim_L v$, but the converse is not always true.  
%Thus $M$ remembers more than it needs to, distinguishing some pairs of strings that are actually equivalent with respect to $L$.  We illustrate this in 
Figure~\ref{fig:breakup}.

\begin{figure}
\begin{center}
\includegraphics[width=1.5in]{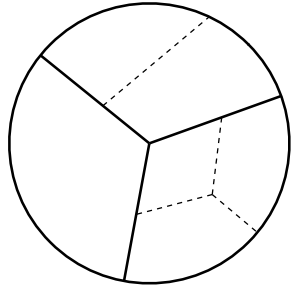}
\end{center}
\caption{The equivalence classes of a language $L$ where $\sim_L$ has three equivalence classes (bold).  A non-minimal DFA $M$ with six states that recognizes $L$ corresponds to a finer equivalence relation $\sim _M$ with smaller equivalence classes (dashed).  It remembers more than it needs to, distinguishing strings $u, v$ that it would be all right to merge.}
\label{fig:breakup}
\end{figure}

Let's say that two states $s, s'$ of $M$ are equivalent if the corresponding equivalence classes of $\sim_M$ lie in the same equivalence class of $\sim_L$.  In that case, if we merge $s$ and $s'$ in $M$'s state space, we get a smaller DFA that still recognizes $L$.  We can obtain the minimal DFA by merging equivalent states until each equivalence class of $\sim_L$ corresponds to a single state.  This yields an algorithm for finding the minimal DFA which runs in polynomial time as a function of the number of states.

The Myhill-Nerode Theorem may seem a little abstract, but it is perfectly concrete.  Doing the following exercise will give you a feel for it if you don't have one already:
\begin{exercise}
Describe the equivalence classes of the three languages from Exercise~\ref{ex:dfa}.  Use them to give the minimal DFA for each language, or prove that the DFA you designed before is minimal.
\end{exercise}

\loesung{
\begin{enumerate}
\item The set of strings in $\{a,b\}^*$ with an even number of $b$s: there are two equivalence classes, namely strings where the number of $b$s is even or odd respectively.
\item The set of strings in $\{a,b,c\}^*$ with no $c$ anywhere to the left of an $a$: there are three equivalence classes, namely strings where we haven't had a $c$, those where we have had a $c$ but no $a$ after it (which are still in the language), and those outside the language.
\item The set of strings in $\{0,1\}^*$ that encode multiples of $3$: interpret the empty string $\eps$ as the number zero. There are three equivalence classes, namely integers whose value mod $3$ is $0$, $1$, or $2$.
\end{enumerate}
Each of the DFAs shown above are minimal, with a single state for each equivalence class.
}

We can also answer some of the questions we raised before about whether we really need as many states as our constructions above suggest.

\begin{exercise}
Describe an infinite family of pairs of languages $(L_p,L_q)$ such that the minimal DFA for $L_p$ has $p$ states, the minimal DFA for $L_q$ has $q$ states, and the minimal DFA for $L_p \cap L_q$ has $pq$ states. 
\end{exercise}

\loesung{
Consider a one-symbol alphabet, $A=\{a\}$.  Let $L_p$ be the set of strings whose length $\ell$ is a multiple of $p$.  It can be recognized by a DFA with $p$ states $\{0,\ldots,p-1\}$, with $\sinit=0$, $\Syes=\{0\}$, and $\delta(i,a) = (i+1) \bmod p$.  Each of these states corresponds to a different equivalence class $[a^i]$, namely those whose length mod $p$ is $i$.  These are inequivalent since if $u \in [a^i]$ and $v \in [a^j]$ where $i \ne j$, then $uw \in L$ but $vw \notin L$ if $w=a^{p-i}$.  Thus this DFA is minimal.

Define $L_q$ similarly; then its minimal DFA has $q$ states.  Their intersection is 
\[
L_p \cap L_q = L_{\lcm(p,q)} \, , 
\]
since $\ell$ is a multiple of both $p$ and $q$ if and only if it is a multiple of their lowest common multiple $\lcm(p,q)$.  By the same argument, the minimal DFA for this language has $\lcm(p,q)$ states.  But if $p$ and $q$ are relatively prime, $\lcm(p,q) = pq$.
}

\begin{exercise}
Describe a family of languages $L_k$, one for each $k$, such that $L_k$ can be recognized by an NFA with $O(k)$ states, but the minimal DFA that recognizes $L_k$ has at least $2^k$ states.  Hint: consider the NFA $M_3$ defined in Section~\ref{sec:nfa} above.
\end{exercise}

\loesung{
Let $A=\{a,b\}$, and let $L_k$ be the set of words whose $k$th-to-last symbol is a $b$.  It can be recognized by an NFA with $k+1$ states analogous to $M_3$.  This NFA can stay in its initial state as long as it wants, but when reading a $b$ it can also move to the first state in a chain of $k$ states.  After that, every transition leads down the chain, and only the last state is accepting.  

We will prove that the minimal DFA that recognizes $L_k$ has at least $2^k$ states.  Let $s, s'$ be two distinct words of length $k$.  If $s=s_1 s_2 \cdots s_{k}$ and $s' = s'_1 s'_2 \cdots s'_{k}$, there is some $t$ such that $s_t \ne s'_t$.  Without loss of generality, say that $s_t = b$ and $s'_t = a$.  Then if $w$ has length $t-1$, $sw \in L$ but $s'w \notin L$.  Thus $s \not\sim s'$, and the number of equivalence classes is at least $2^k$, the number of distinct strings of length $k$.

This is a long-winded way of saying what we know intuitively, namely that we have to keep track of the last $k$ symbols to tell whether the $k$th-to-last symbol is a $b$.  This information is both necessary and sufficient.  In other words, there are exactly $2^k$ equivalence classes: for each string $s$ of length $k$, its equivalence class $[s]$ is the set of strings that end in $s$.  (We treat words of length less than $k$ by pretending that they start with $a$s.)  Thus the minimal DFA for $L_k$ has exactly $2^k$ states.
}

The following exercises show that even reversing a language, or concatenating two languages, can greatly increase the number of states.  Hint: the languages $L_k$ from the previous exercise have many uses.

\begin{exercise}%[Some languages are harder backwards than forwards.]
\label{ex:reverse-harder}
Describe a family of languages $L_k$, one for each $k$, such that $L_k$ can be recognized by a DFA with $O(k)$ states, but the minimal DFA for $L_k^R$ has at least $2^k$ states.
\end{exercise}%

\loesung{
If $L_k$ is the set of strings whose $k$th-to-last symbol is a $b$, then $L_k^R$ is the set of strings whose $k$th symbol is a $b$.  We can recognize $L_k^R$ with a DFA with $k+2$ states, namely a chain of $k$ states ending in a branch between an accepting and rejecting state.
}

\begin{exercise}%[Concatenation can be tricky.]
Describe a family of pairs of languages $(L_1,L_2)$, one for each $k$, such that $L_1$ can be recognized by a DFA with a constant number of states and $L_2$ can be recognized by a DFA with $k$ states, but the minimal DFA for $L_1 L_2$ has\, $\Omega(2^k)$ states.
\end{exercise}

\loesung{
Let $L_1 = \{a,b\}^*$ and let $L_2 = b \{a,b\}^{k-1}$, the set of strings of length $k$ that start with a $b$.  Then $L_k = L_1 L_2$.
}

The sapient reader will wonder whether there is an analog of the Myhill-Nerode Theorem for NFAs, and whether the minimal NFA for a language has a similarly nice description.  It turns out that finding the minimal NFA is much harder than finding the minimal DFA.  Rather than being in P, it is PSPACE-complete (see Chapter 8).  That is, it is among the hardest problems that can be solved with a polynomial amount of memory.

\section{Nonregular Languages}
\label{sec:nonreg}

The Myhill-Nerode Theorem has another consequence.  Namely, it tells us exactly when a language is regular:

\begin{corollary}
A language $L$ is regular if and only if\; $\sim_L$ has a finite number of equivalence classes.
\end{corollary}

Thus to prove that a language $L$ is not regular---that no DFA, no matter how many states it has, can recognize it---all we have to do is show that $\sim_L$ has an infinite number of equivalence classes.  This may sound like a tall order, but it's quite easy.  We just need to exhibit an infinite set of strings $u_1, u_2, \ldots$ such that $u_i \not\sim u_j$ for any $i \ne j$.  And to prove that $u_i \not\sim u_j$, we just need to give a string $w$ such that $u_i w \in L$ but $u_j w \notin L$ or vice versa.

For example, given a string $w$ and a symbol $a$, let $\#_a(w)$ denote the number of $a$s in $w$.  Then consider the following language.: 
\[
L_{a=b} = \left\{ w \in \{a,b\}^* \mid \#_a(w) = \#_b(w) \right\} \, . 
%\mid \mbox{$w$ has an equal number of $a$s and $b$s} \right\} \, . 
\]
Intuitively, in order to recognize this language we have to be able to count the $a$s and $b$s, and to count up to any number requires an infinite number of states.  Our definition of equivalence classes lets us make this intuition rigorous.  Consider the set of words $\{ a^i \mid i \ge 0 \} = \{ \eps, a, aa, aaa, \ldots \}$.  If $i \ne j$, then $a^i \not\sim a^j$ since 
\[
a^i b^i \in L_{a=b} \quad \text{but} \quad a^j b^i \notin L_{a=b} \, . 
\]
Thus each $i$ corresponds to a different equivalence class. Any DFA with $n$ states will fail to recognize $L$ since it will confuse $a^i$ with $a^j$ for sufficiently large $i, j$.  The best a DFA with $n$ states can do is count up to $n$.

\begin{exercise}
Describe all the equivalence classes of $L_{a=b}$, starting with $[a^i]$.
\end{exercise}

\loesung{
There is one equivalence class for each integer $z \in Z$.  If we write $[a^i] = [i]$, $[b^i] = [-i]$, and $[\eps]=[0]$, then $\delta([i],a) = [i+1]$ and $\delta([i],b) = [i-1]$.
}

Any automaton that recognizes $L_{a=b}$ has to have an infinite number of states.  Figure~\ref{fig:ab} shows an infinite-state automaton that does the job.  The previous exercise shows that this automaton is the ``smallest possible,'' in the sense that each equivalence class corresponds to a single state.  Clearly this automaton, while infinite, has a simple finite description---but not a description that fits within the framework of DFAs or NFAs.

\begin{figure}
\begin{center}
\includegraphics[width=2.5in]{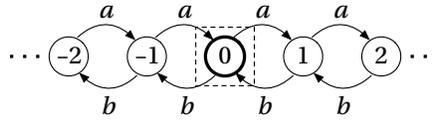}
\end{center}
\caption{The smallest possible infinite-state automaton (yes, that makes sense) that recognizes the non-regular language $L_{a=b}$.}
\label{fig:ab}
\end{figure}

\begin{exercise}
\label{ex:abcd}
Consider the language 
\[
L_{a=b,c=d} = \left\{ w \in \{a,b,c,d\}^* \mid \#_a(w) = \#_b(w) \text{ and } \#_c(w) = \#_d(w) \right\} \, . 
\]
What are its equivalence classes?  What does its minimal infinite-state machine look like?
\end{exercise}

\loesung{
For each $(x,y) \in \Z^2$, there is an equivalence class of words where 
\[ 
\#_a(w) - \#_b(w) = x \quad \text{and} \quad \#_c(w) - \#_d(w) = y \, . 
\]
These classes, and therefore the states of its minimal machine, form a grid.  The initial state, and only accepting state, is at the origin $(0,0)$, and reading $a, b, c$, or $d$ causes the machine to move right, left, up, or down respectively.
}

\begin{exercise}%[Brackets.]
\label{ex:dyck}
The Dyck language $D_1$ is the set of strings of properly matched left and right parentheses, 
\[ 
D_1=\{ \eps, (), ()(), (()), ()()(), (())(), ()(()), (()()), ((())), \ldots\ \} \, . 
\]
Prove, using any technique you like, that $D_2$ is not regular.  What are its equivalence classes?  What does its minimal infinite-state machine look like?

Now describe the equivalence classes for the language $D_2$ with two types of brackets, round and square.  These must be nested properly, so that $[()]$ is allowed but $[(])$ is not.  Draw a picture of the minimal infinite-state machine for $D_2$ in a way that makes its structure clear.  How does this picture generalize to the language $D_k$ where there are $k$ types of brackets? 
\end{exercise}

\loesung{
(Sketch) $D_1$ has an equivalence class for each $i \in \N$, consisting of the set of strings that are legal so far, but which have $i$ unmatched left parentheses.  These are all distinct, since $(^i )^i \in D_1$ but $(^j )^i \notin D_1$ if $i \ne j$.

The equivalence classes for $D_2$ are more complicated.  Assuming the word is legal so far, there is a sequence of left brackets, round and square, waiting to be matched.  Each such sequence corresponds to a different equivalent class, so the infinite-state machine has a state for each sequence.  It is shaped like a tree: each sequence $\sigma$ has two children, $\sigma($ and $\sigma[$, which we move to if we read $($ or $[$.  If we read a right bracket, we move to $\sigma$'s parent if it matches $\sigma$'s last bracket; otherwise we move to a permanently rejecting state.
}

\section{The Pumping Lemma}
\label{sec:pumping}

The framework of equivalence classes is by far the most simple, elegant, and fundamental way to tell whether or not a language is regular.  But there are other techniques as well.  Here we describe the Pumping Lemma, which states a necessary condition for a language to be regular.  It is not as useful or as easy to apply as the Myhill-Nerode Theorem, but the proof is a nice use of the pigeonhole principle, and applying it gives us some valuable exercise in juggling quantifiers.  It states that any sufficiently long string in a regular language can be ``pumped,'' repeating some middle section of it as many times as we like, and still be in the language.

As before we use $|w|$ to denote the length of a string $w$.  We use $y^t$ to denote $y$ concatenated with itself $t$ times.
\begin{lemma}
\label{lem:pumping}
Suppose $L$ is a regular language.  Then there is an integer $p$ such that any $w \in L$ with $|w| \ge p$ can be written as a concatenation of three strings, $w=xyz$, where
\begin{enumerate}
\item $|xy| \le p$, 
\item $|y| > 0$, i.e., $y \ne \eps$, and
\item for all integers $t \ge 0$, $xy^tz \in L$. 
\end{enumerate}
\end{lemma}

\begin{proof}
As the reader may have guessed, the constant $p$ is the number of states in the minimal DFA $M$ that recognizes $L$.  Including the initial state $\sinit$, reading the first $p$ symbols of $w$ takes $M$ to $p+1$ different states.  By the pigeonhole principle, two of these states must be the same, which we denote $s$.  Let $x$ be the part of $w$ that takes $M$ to $s$ for the first time, let $y$ be the part of $w$ that brings $M$ back to $s$ for its first return visit, and let $z$ be the rest of $w$ as shown in Figure~\ref{fig:pumping}.  Then $x$ and $y$ satisfy the first two conditions, and 
\[
\delta^*(\sinit,x) = \delta^*(\sinit,xy) = \delta^*(s,y) = s \, .
\]
By induction on $t$, for any $t \ge 0$ we have $\delta^*(\sinit,xy^t) = s$, and therefore $\delta^*(\sinit,xy^tz) = \delta^*(s,z)$.
%\[
%\delta^*(\sinit,xy^tz) 
%%= \delta^*(\delta^*(\sinit,xy^t),z) 
%= \delta^*(s,z) 
%%= \delta^*(\delta^*(\sinit,xy),z) = 
%%= \delta^*(\sinit,xyz) = \delta^*(\sinit,w) 
%\, . 
%\]
In particular, since $w = xyz \in L$ we have $\delta^*(\sinit,xyz) \in \Syes$, so $xy^tz \in L$ for all $t \ge 0$.
\end{proof}

\begin{figure}
\begin{center}
\includegraphics[width=2in]{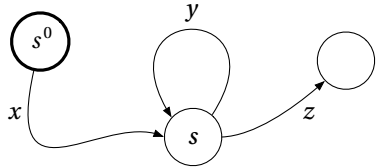}
\end{center}
\caption{If a DFA has $p$ states, the first $p$ symbols of $w$ must cause it to visit some state $s$ twice.  We let $x$ be the part of $w$ that first takes $M$ to $s$, let $y$ be the part of $w$ that brings $M$ back to $s$ for the first return visit, and let $z$ be the rest of $w$.  Then for any $t$, the word $xy^tz$ takes $M$ to the same state that $w=xyz$ does.}
\label{fig:pumping}
\end{figure}

Note that the condition described by the Pumping Lemma is necessary, but not sufficient, for $L$ to be regular.  In other words, while it holds for any regular language, it also holds for some non-regular languages.  Thus we can prove that $L$ is not regular by showing that the Pumping Lemma is violated, but we cannot prove that $L$ is regular by showing that it is fulfilled.

Logically, the Pumping Lemma consists of a nested series of quantifiers.  Let's phrase it in terms of $\exists$ (there exists) and $\forall$ (for all).  If $L$ is regular, then
\begin{tabbing}
$\exists$ an integer $p$ such that \\ 
\qquad $\forall w \in L$ with $|w| \ge p$, \\
\qquad \qquad $\exists x, y, z$ such that $w=xyz$, $|xy| \le p$, $y \ne \eps$, and \\
\qquad \qquad \qquad $\forall$ integers $t \ge 0$, $xy^tz \in L$.
\end{tabbing}
Negating all this turns the $\exists$s into $\forall$s and vice versa.  Thus if you want to use the Pumping Lemma to show that $L$ is not regular, you need to show that 
\begin{tabbing}
$\forall$ integers $p$, \\ 
\qquad $\exists w \in L$ with $|w| \ge p$ such that \\
\qquad \qquad $\forall x, y, z$ such that $w=xyz$, $|xy| \le p$, and $y \ne \eps$, \\
\qquad \qquad \qquad $\exists$ an integer $t \ge 0$ such that $xy^tz \notin L$.
\end{tabbing}

You can think of this kind of proof as a game.  You are trying to prove that $L$ is not regular, and I am trying to stop you.  The $\forall$s are my turns, and the $\exists$s are your turns.  I get to choose the integer $p$.  No matter what $p$ I choose, you need to be able to produce a string $w \in L$ of length at least $p$, such that no matter how I try to divide it into a beginning, middle, and end by writing $w=xyz$, you can produce a $t$ such that $xy^tz \notin L$.  

If you have a winning strategy in this game, then the Pumping Lemma is violated and $L$ is not regular.  But it's not enough, for instance, for you to give an example of a word $w$ that can't be pumped---you have to be able to give such a word which is longer than any $p$ that I care to name.

Let's illustrate this by proving that the language
\[
L = \{ a^n b^n \mid n \ge 0 \} = \{ \eps, ab, aabb, aaabbb, \ldots\ \} 
\]
is not regular.  This is extremely easy with the equivalence class method, but let's use the Pumping Lemma instead.  First I name an integer $p$.  You then reply with $w = a^p b^p$.  No matter how I try to write it as $w=xyz$, the requirement that $|xy| \le p$ means that both $x$ and $y$ consist of $a$s.  In particular, $y=a^q$ for some $q > 0$, since $y \neq \eps$.  But then you can take $t=0$, and point out that $xz = a^{p-q} b^q \notin L$.  Any other $t \ne 1$ works equally well.

\begin{exercise}
Prove that each of these languages is nonregular, first using the equivalence class method and then using the Pumping Lemma.  
%Note how much easier the equivalence class method is to apply.
\begin{enumerate}
\item $\{a^{2^i} \mid i \ge 0\} = \{a, aa, aaaa, aaaaaaaa, \ldots \}$.

\loesung{
Let $u=a^{2^i}$ and $v=a^{2^j}$ for $i \ne j$.  Then $uu = a^{2^{i+1}} \in L$ but $vu = a^{2^i+2^j} \notin L$, since $2^i+2^j$ is not a power of $2$ unless $i=j$.  Thus $u \not\sim v$.  These are not all the equivalence classes, but they are enough to show that there are an infinite number, and therefore that $L$ is not regular.

Alternately, let $p$ be any constant, and choose $w = a^{2^i}$ so that $2^i > p$.  If $w=xyz$ with $|xy| \le p$ and $|y| > 0$, then $y=a^t$ for some $0 < t \le p < 2^i$.  But then $xy^2z = a^{2^i+t} \notin L$, since $2^i+t$ is not a power of $2$.  This violates the pumping lemma, so $L$ is not regular.
}

\item The language $\Lpal$ of palindromes over a two-symbol alphabet, i.e., $\big\{w \in \{a,b\}^* \mid w=w^R \big\}$.

\loesung{
Let $u=a^i b$ and $v=a^j b$ for $i \ne j$.  Then $u \not\sim v$, since $ua^i \in L$ but $va^i \notin L$.  Thus there are an infinite number of equivalence classes, and $L$ is not regular.

Alternately, given a constant $p$, let $w=a^p b a^p$.  If $w=xyz$ with $|xy| \le p$ and $|y| > 0$, then $y=a^t$ for some $t > 0$ and occurs in the first half of $w$, and $xz = a^{p-t} b a^t \notin L$.  This violates the pumping lemma, so $L$ is not regular.
}

\item The language $\Lcopy$ of words repeated twice, $\big\{ ww \mid w \in \{a,b\}^* \big\}$.

\loesung{
Let $u=a^i b$ and $v=a^j b$ for $i \ne j$.  Then $u \not\sim v$ since $uu \in L$ but $vu \notin L$.  Thus there are an infinite number of equivalence classes, and $L$ is not regular.

Alternately, given a constant $p$, let $w=a^p b a^p b$.  If $w=xyz$ with $|xy| \le p$ and $|y| > 0$, then $y=a^t$ for some $t > 0$ and occurs in the first half of $w$, and $xz = a^{p-t} b a^p b \notin L$.  This violates the pumping lemma, so $L$ is not regular.
}

\item $\{ a^i b^j \mid i > j \ge 0\}$.

\loesung{
Let $u=a^i$ and $v=a^k$ for $i \ne k$.  Assume without loss of generality that $i > k$.  Then $u b^k = a^i b^k \in L$ but $v b^k = a^k b^k \notin L$.  Thus there are an infinite number of equivalence classes, and $L$ is not regular.

Alternately, given a constant $p$, let $w=a^p b^{p-1}$.  If $w=xyz$ with $|xy| \le p$ and $|y| > 0$, then $y=a^t$ for some $t > 0$ and occurs in the first half of $w$, and $xz = a^{p-t} b^{p-1} \notin L$.  This violates the pumping lemma, so $L$ is not regular.
}
\end{enumerate}
\end{exercise}

\begin{exercise}
Given a language $L$, the language $\sort(L)$ consists of the words in $L$ with their characters sorted in alphabetical order.  For instance, if 
\[
L=\{bab, cca, abc\}
\]
then
\[
\sort(L) = \{abb,acc,abc\} \, .
\]
Give an example of a regular language $L_1$ such that $\sort(L_1)$ is nonregular, and a nonregular language $L_2$ such that $\sort(L_2)$ is regular.  You may use any technique you like to prove that the languages are nonregular.

\end{exercise}

\section{Regular Expressions}
\label{sec:regexp}

In a moment, we will move on from finite to infinite-state machines, and define classes of automata that recognize many of the non-regular languages described above.  But first, it's worth looking at one more characterization of regular languages, because of its elegance and common use in the real world.  A \emph{regular expression} is a parenthesized expression formed of symbols in the alphabet $A$ and the empty word, combined with the operators of concatenation, union (often written $+$ instead of $\cup$) and the Kleene star $*$.  Each such expression represents a language.  For example, 
\[
(a+ba)^* \,(\eps+b) 
\]
represents the set of strings generated in the following way: as many times as you like, including zero, print $a$ or $ba$.  Then, if you like, print $b$.  A moment's thought shows that this is our old friend $\nobb$.  There are many other regular expressions that represent the same language, such as
\[
(\eps+b) \,(aa^* \,b)^* \,a^* \, . 
\]

We can define regular expressions inductively as follows.
\begin{enumerate}
\item $\emptyset$ is a regular expression.
\item $\eps$ is a regular expression.
\item Any symbol $a \in A$ is a regular expression.
\item If $\phi_1$ and $\phi_2$ are regular expressions, then so is $\phi_1 \phi_2$.
\item If $\phi_1$ and $\phi_2$ are regular expressions, then so is $\phi_1 + \phi_2$.
\item If $\phi$ is a regular expression, then so is $\phi^*$.
\end{enumerate}
In case it isn't already clear what language a regular expression represents, 
\begin{align*}
L(\emptyset) &= \emptyset \\
L(\eps) &= \{\eps\} \\
L(a) &= \{a\} \ \mbox{for any $a \in A$} \\
L(\phi_1 \phi_2) &= L(\phi_1) \,L(\phi_2) \\
L(\phi_1 + \phi_2) &= L(\phi_1) \cup L(\phi_2) \\
L(\phi^*) &= L(\phi)^* \, . 
\end{align*}

Regular expressions can express exactly the languages that DFAs and NFAs recognize.  In one direction, the proof is easy:

\begin{theorem}
If a language can be represented as a regular expression, then it is regular.
\end{theorem}

\begin{proof}
This follows inductively from the fact that $\emptyset$, $\{\eps\}$ and $\{a\}$ are regular languages, and that the regular languages are closed under concatenation, union, and $*$.  
\end{proof}

The other direction is a little harder:

\begin{theorem}
If a language is regular, then it can be represented as a regular expression.
\end{theorem}

\begin{proof}%[Proof sketch]
We start with the transition diagram of an NFA that recognizes $L$.  We allow each arrow to be labeled with a regular expression instead of just a single symbol.  We will shrink the diagram, removing states and edges and updating these labels in a certain way, until there are just two states left with a single arrow between them.

First we create a single accepting state $\saccept$ by drawing $\eps$-transitions from each $s \in \Syes$ to $\saccept$.  We then reduce the number of states as follows.  Let $s$ be a state other than $\sinit$ and $\saccept$.  We can remove $s$, creating new transitions between its neighbors.  For each pair of states $u$ and $v$ with arrows leading from $u$ to $s$ and from $s$ to $v$, labeled with $\phi_1$ and $\phi_2$ respectively, we create an arrow from $u$ to $v$ labeled with with $\phi_1 \phi_2$.  If $s$ had a self-loop labeled with $\phi_3$, we label the new arrow with $\phi_1 \phi_3^* \phi_2$ instead.  

We also reduce the number of edges as follows.  Whenever we have two arrows pointing from $u$ to $v$ labeled with expressions $\phi_1$ and $\phi_2$ respectively, we replace them with a single arrow labeled with $\phi_1+\phi_2$.  Similarly, if $u$ has two self-loops labeled $\phi_1$ and $\phi_2$, we replace them with a single self-loop labeled $\phi_1+\phi_2$.

We show these rules in Figure~\ref{fig:gfa}.  The fact that they work follows from our definition of the language represented by a regular expression.  A path of length two gets replaced with $\phi_1 \phi_2$ since we go through one arrow and then the other, a loop gets replaced with $\phi^*$ since we can go around it any number of times, and a pair of arrows gets replaced with $\phi_1+\phi_2$ since we can follow one arrow or the other.  

After we have reduced the diagram to a single arrow from $\sinit$ to $\saccept$ labeled $\phi$, with perhaps a self-loop on $\sinit$ labeled $\phi_0$, then the regular expression for the language is $\phi_0^* \phi$.  If there is no such self-loop, then the regular expression is simply $\phi$.
\end{proof}

\begin{figure}
\begin{center}
\includegraphics[width=2.8in]{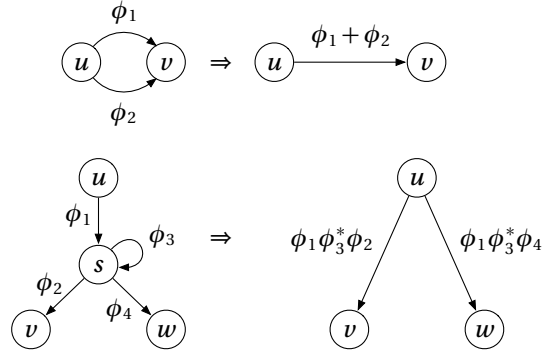}
\end{center}
\caption{Rules for reducing the size of an NFA's transition diagram while labeling its arrows with regular expressions.  When the entire diagram has been reduced to $\sinit$ and $\saccept$ with a single arrow between them, its label is the regular expression for $L$.}
\label{fig:gfa}
\end{figure}

\begin{exercise}
Apply this algorithm to the DFA for $\nobb$, and for the three languages in Exercise~\ref{ex:dfa}.  Note that there are usually multiple orders in which you can remove states from the diagram.  Use these to produce two different regular expressions for each of these languages.
\end{exercise}

Like DFAs and NFAs, regular expressions make it easier to prove certain things about regular languages.  

\begin{exercise}
Recall the definition of $L^R$ from Exercise~\ref{ex:reverse}.  Give a simple inductive proof using regular expressions that $L$ is regular if and only if $L^R$ is regular.
\end{exercise}

\loesung{
Base cases: $\emptyset^R = \emptyset$ is regular, $\{\eps\}^R = \{\eps\}$ is regular, and $\{a\}^R = \{a\}$ is regular for any $a \in A$.

Inductive steps: suppose that $L$ is described by a regular expression $\phi$, which is defined from smaller regular expressions $\phi_1, \phi_2$ by concatenation, union, or $*$.  Assume by induction that $\phi_1^R$ and $\phi_2^R$ are regular.  We can define $\phi^R$ inductively as follows, 
\[
(\phi_1 \phi_2)^R = \phi_2^R \phi_1^R \, , \quad 
(\phi_1 + \phi_2)^R = \phi_1^R + \phi_2^R \, , \quad \text{and} \quad
(\phi_1^*)^R = (\phi_1^R)^* \, . 
\]
The corresponding languages are regular, since the regular languages are closed under concatenation, union, and $*$.  Thus $L^R = L(\phi^R)$ is regular.

This shows that if $L$ is regular, then $L^R$ is regular.  But since $(L^R)^R = L$, the converse is true as well, so $L$ is regular if and only if $L^R$ is.
}

The reader might wonder why we don't allow other closure operators, like intersection or complement, in our definition of regular expressions.  In fact we can, and these operators can make regular expressions exponentially more compact.  In practice we usually stick to concatenation, union, and $*$ because there are efficient algorithms for searching a text for strings matched by expressions of this form.

Note that unlike DFAs and NFAs, regular expressions do not give an algorithm for recognizing a language, taking a word $w$ as input and saying ``yes'' or ``no'' if $w \in L$ or not.  Instead, they \emph{define} a language, creating it ``all at once'' as a set.  Below we will see yet another approach to languages---a grammar that generates words, building them from scratch according to simple rules.  But first, let's step beyond DFAs and NFAs, and look at some simple kinds of infinite-state machines.

\section{Counter Automata}

We saw earlier that some languages require an infinite-state machine to recognize them.  Of course, \emph{any} language can be recognized by some infinite-state machine, as the following exercise shows:
\begin{exercise}
Show that any language $L \subseteq A^*$ can be recognized by a machine with a countably infinite set of states.  Hint: consider an infinite tree where each node has $|A|$ children.
\end{exercise}
\noindent
But the vast majority of such machines, like the vast majority of languages, have no finite description.  Are there reasonable classes of infinite-state machines whose state spaces are structured enough for us to describe them succinctly?

A common way to invent such machines is to start with a finite-state automaton and give it access to some additional data structure.  For instance, suppose we give a DFA access to a \emph{counter}: a data structure whose states correspond to nonnegative integers.  To keep things simple, we will only allow the DFA to access and modify this counter in a handful of ways.  Specifically, the DFA can only update the counter by incrementing or decrementing its value by $1$.  And rather than giving the DFA access to the counter's value, we will only allow it to ask whether it is zero or not.  

Let's call this combined machine a \emph{deterministic one-counter automaton}, or $1$-DCA for short.  We can represent it in several ways.  One is with a transition function of the form
\[
\delta : S \times \{ \textrm{zero}, \textrm{nonzero} \} \times A \to S \times \{ \textrm{inc}, \textrm{dec}, \textrm{do\ nothing} \} \, . 
\]
This function takes the current state of the DFA, the zeroness or nonzeroness of the counter, and an input symbol.  It returns the new state of the DFA and an action to perform on the counter.  As before, we specify an initial state $\sinit \in S$.  For simplicity, we take the initial value of the counter to be zero.  

The state space of a $1$-DCA looks roughly like Figure~\ref{fig:ab}, although there the counter takes both positive and negative values.  More generally, the state space consists of a kind of product of the DFA's transition diagram $S$ and the natural numbers $\N$, with a state $(s,n)$ for each pair $s \in S$, $n \in \N$.  It accepts a word if its final state is in some accepting set $\Syes$.  However, to allow its response to depend on the counter as well as on the DFA, we define $\Syes$ as a subset of $S \times \{ \textrm{zero}, \textrm{nonzero} \}$.

We have already seen several non-regular languages that can be recognized by a one-counter automaton, such as the language $L_{a=b}$ of words with an equal number of $a$s and $b$s.  
%(Exercise: write a formal description of a D1CA for this language.)  
Thus counter automata are more powerful than DFAs.  On the other hand, a language like $\Lpal$, the set of palindromes, seems to require much more memory than a one-counter automaton possesses.  Can we bound the power of counter automata, as we bounded the power of finite-state automata before?  

Let's generalize our definition to automata with $k$ counters, allowing the DFA to increment, decrement, and check each one for zero, and call such things $k$-DCAs.  As in Exercise~\ref{ex:abcd}, the state space of such a machine is essentially a $k$-dimensional grid.  The following theorem uses the equivalence class machinery we invented for DFAs.  It shows that, while the number of equivalence classes of a language accepted by a $k$-DCA may be infinite, it must grow in a controlled way as a function of the length of the input.

\begin{theorem}
\label{thm:dca}
Let $M$ be a $k$-DCA.  In its first $t$ steps, the number of different states it can reach is $\bigo(t^k)$, i.e., at most $C t^k$ for some constant $C$ (where $C$ depends on $M$ but not on $t$).  Therefore, if $M$ recognizes a language $L$, it must be the case that $\sim_L$ has $O(t^k)$ different equivalence classes among words of length $t$.
\end{theorem}

\begin{exercise}
Prove this theorem.  
\end{exercise}

\loesung{
In the first $t$ steps, no counter can increase in value beyond $t$.  Thus the total number of states $M$ can be in is at most $|S| t^k$.
}

Note that these automata are required to run in ``real time,'' taking a single step for each input symbol and returning their answer as soon as they reach the end of the input.  Equivalently, there are no $\eps$-transitions.  As we will see in Section~7.6, relaxing this requirement makes even two-counter automata capable of universal computation.

Theorem~\ref{thm:dca} lets us prove pretty much everything we might want to know about counter automata.  For starters, the more counters we have, the more powerful they are:

\begin{exercise}
Describe a family of languages $L_k$ for $k = 0, 1, 2 \ldots$ that can be recognized by a $k$-DCA but not by a $k'$-DCA for any $k' < k$.
\end{exercise}

\loesung{
(Sketch) Let $A=\{a_1,b_1,a_2,b_2,\ldots,a_k,b_k\}$, and define $L$ as the set of words $w$ where, for each $1 \le i \le k$, $\#_{a_i}(w) = \#_{b_i}(w)$.  There is a distinct equivalence class for each vector $(z_1,\ldots,z_k) \in \Z^k$.  If $w$ is of length $t$, we know that $\sum_i |z_i| \le t$; this gives about $2^{k-1} t^k / k!$ possible states (exercise), which is not $O(t^{k'})$ for any $k' < k$.
}

The next exercise confirms our intuition that we need $k_1+k_2$ counters to run a $k_1$-DCA and a $k_2$-DCA in parallel.

\begin{exercise}
Show that for any integers $k_1, k_2 \ge 0$, there are languages $L_1, L_2$ that can be recognized by a $k_1$-DCA and a $k_2$-DCA respectively, such that $L_1 \cap L_2$ and $L_1 \cup L_2$ cannot be recognized by a $k$-DCA for any $k < k_1+k_2$.
\end{exercise}

\loesung{
(Sketch) Use the same approach as the previous exercise, with $k=k_1+k_2$.  Let $L_1$ require that $\#_{a_i}(w) = \#_{b_i}(w)$ for $1 \le i \le k_1$, and let $L_2$ require this for $k_1 < i \le k$.
}

On the other hand, if we lump $k$-DCAs together for all  $k$ into a single complexity class, it is closed under all Boolean operations:
\begin{exercise}
Say that $L$ is \emph{deterministic constant-counter} if it can be recognized by a $k$-DCA for some constant $k$.  Show that the class of such languages is closed under intersection, union, and complement.
\end{exercise}

We can similarly define nondeterministic counter automata, or NCAs.  As with NFAs, we say that an NCA accepts a word if there exists a computation path leading to an accepting state.  NCAs can be more powerful than DCAs with the same number of counters:
\begin{exercise}
Consider the language
\[
L_{a=b \vee b=c} = \left\{ w \in \{a,b,c\}^* \mid \#_a(w) = \#_b(w) \mbox{ or } \#_b(w) = \#_c(w)\right\} \, . 
\]
Show that $L_{a=b \vee b=c}$ can be recognized by a $1$-NCA, but not by a $1$-DCA.
\end{exercise}

\loesung{
(Sketch) A DCA needs to keep track of both $\#_a(w)-\#b(w)$ and $\#_b(w)-\#_c(w)$.  For words of length $t$, this gives $\Omega(t^2)$ equivalence classes.  An NCA, on the other hand, can guess at the outset which of these to keep track of, and devote one counter to it.
}

Note how the $1$-NCA for this language uses nondeterminism in an essential way, choosing whether to compare the $a$s with the $b$s, or the $b$s with the $c$s.

The next exercise shows that even $1$-NCAs cannot be simulated by DCAs with any constant number of counters.  Thus unlike finite-state automata, for counter automata, adding nondeterminism provably increases their computational power.

\begin{exercise}
\label{ex:pal-ca}
Recall that $\Lpal$ is the language of palindromes over a two-symbol alphabet.  Show that 
\begin{enumerate}
\item $\Lpal$ is not a deterministic constant-counter language, but  
\item its complement $\overline{\Lpal}$ can be recognized by a $1$-NCA.
\end{enumerate}
Conclude that $1$-NCAs can recognize some languages that cannot be recognized by a $k$-DCA for any $k$.
\end{exercise}

\loesung{
(Sketch) Every word $w$ is its own equivalence class in $\Lpal$ (exercise).  This gives $2^t$ equivalence classes among words of length $t$, which is not $O(t^k)$ for any constant $k$.  Since the DCA languages are closed under complement, $\Lpal$ cannot be recognized by a $k$-DCA for any $k$.  

On the other hand, a 1-NCA can recognize $\overline{\Lpal}$ by guessing a $t$ such that the $t$th symbol differs from the $t$th-to-last symbol.  It increments the counter until it nondeterministically chooses to stop at $t$, and records $w_t$.  At some later time $s$ where $w_s \ne w_t$, it nondeterministically chooses to start counting down again.  If the counter is zero at the end of the word, it accepts.
}

At the moment, it's not clear how to prove that a language cannot be recognized by a $1$-NCA, let alone by a $k$-NCA.  Do you have any ideas?

\section{Stacks and Push-Down Automata}

Let's continue defining machines where a finite-state automaton has access to a data structure with an infinite, but structured, set of states.  One of the most well-known and natural data structures is a \emph{stack}.  At any given point in time, it contains a finite string written in some alphabet.  The user of this data structure---our finite automaton---is allowed to check to see whether the stack is empty, and to look at the top symbol if it isn't.  It can modify the stack by ``pushing'' a symbol on top of the stack, or ``popping'' the top symbol off of it.  

A stack works in ``last-in, first-out'' order.  Like a stack of plates, the top symbol is the one most recently pushed onto it.  It can store an infinite amount of information, but it can only be accessed in a limited way---in order to see symbols deep inside the stack, the user must pop off the symbols above them.

Sadly, for historical reasons a finite-state automaton connected to a stack is called a \emph{push-down automaton} or PDA, rather than a stack automaton (which is reserved for a fancier type of machine).  We can write the transition function of a deterministic PDA, or DPDA, as follows.  As before, $S$ denotes the DFA's state space and $A$ denotes the input alphabet, and now $\Gamma$ denotes the alphabet of symbols on the stack.
\[
\delta : S \times \left( \Gamma \cup \{ \textrm{empty} \} \right) \times A
\to S \times \left( \{ (\textrm{push}, \gamma) \mid \gamma \in \Gamma \} \cup \{ \textrm{pop} , \textrm{do nothing} \} \right) \, . 
\]
This takes the DFA's current state, the top symbol of the stack or the fact that it is empty, and an input symbol.  It returns a new state and an action to perform on the stack.  

We start in an initial state $\sinit \in S$ and with an empty stack.  Once again, we accept a word if the final state is in some subset $\Syes$.  However, we will often want the criterion for acceptance to depend on the stack being empty, so we define $\Syes$ as a subset of $S \times \left( \Gamma \cup \{ \textrm{empty} \} \right)$.  (This differs minutely from the definition you may find in other books, but it avoids some technical annoyances.)  We denote the language recognized by a DPDA $P$ as $L(P)$.

\begin{figure}
\begin{center}
\includegraphics[width=1.6in]{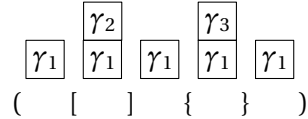}
\end{center}
\caption{A deterministic push-down automaton with three stack symbols, $\Gamma=\{\gamma_1,\gamma_3,\gamma_3\}$, recognizing the word $([]\{\})$ in the language $D_3$ of properly nested and matched strings with three types of brackets.  The stack is empty at the beginning and end of the process.}
\label{fig:dyck}
\end{figure}

The canonical languages recognized by push-down automata are the bracket languages $D_k$ of Exercise~\ref{ex:dyck}, as illustrated in Figure~\ref{fig:dyck}.  Each symbol has a matching partner: the DPDA pushes a stack symbol when it sees a left bracket, and pops when it sees a right bracket.  There are stack symbols $\gamma_1, \ldots, \gamma_k$ for each type of bracket, and we use these symbols to check that the type of each bracket matches that of its partner.  When the stack is empty, all the brackets have been matched, and we can accept.  Due to the last-in, first-out character of the stack, these partners must be nested as in $([])$, rather than crossing as in $([)]$.  If we push $\gamma_1$ and then $\gamma_2$, we have to pop $\gamma_2$ before we can pop $\gamma_1$.  

As you already know if you did Exercise~\ref{ex:dyck}, the state space of a PDA is shaped like a tree.  If $|\Gamma|=k$, each node in this tree has $k$ children.  Pushing the symbol $\gamma_i$ corresponds to moving to your $i$th child, and popping corresponds to moving up to your parent.  At the root of the tree, the stack is empty.

In the case of the bracket language, the DFA has a single state $s$, and $\Syes=\{(s,\textrm{empty})\}$.  With additional states, we can impose regular-language-like constraints, such as demanding that we are never inside more than one layer of curly brackets.  More generally:

\begin{exercise}
Show that the DPDA languages are closed under intersection with regular languages.  That is, if $L$ can be recognized by a DPDA and $R$ is regular, then $L \cap R$ can be recognized by a DPDA.
\end{exercise}

We can define NPDAs in analogy to our other nondeterministic machines, allowing $\delta$ to be multi-valued and accepting if a computation path exists that ends in an accepting state.  

\begin{exercise}
\label{ex:closed-int-reg}
Show that the NPDA languages are also closed under intersection with regular languages.
\end{exercise}

As for counter machines, we will see below that NPDAs are strictly more powerful than DPDAs.  For now, note that NPDAs can recognize palindromes:

\begin{exercise}
Show that $\Lpal$ can be recognized by an NPDA.  Do you think it can be recognized by a DPDA?  How could you change the definition of $\Lpal$ to make it easier for a DPDA?
\end{exercise}

As we will see below, PDAs are incomparable with counter automata---each type of machine can recognize some languages that the other cannot.  For now, consider the following exercises:

\begin{exercise}
Show that a $1$-DCA can be simulated by a DPDA, and similarly for $1$-NCAs and NPDAs.  Do you think this is true for two-counter automata as well?
\end{exercise}

\begin{exercise}
Is $D_2$ a deterministic constant-counter language, i.e., is it recognizable by a $k$-DCA for any constant $k$?
\end{exercise}

As for all our deterministic and nondeterministic machines, the DPDA languages are closed under complement, and the NPDA languages are closed under union.  What other closure properties do you think these classes have?  Which do you think they lack?  

And how might we prove that a language cannot be recognized by a DPDA?  Unlike counter automata, we can't get anywhere by counting equivalence classes.  In $t$ steps, a PDA can reach $|\Gamma|^t$ different states.  If $|\Gamma|=|A|$, this is enough to distinguish any pair of words of length $t$ from each other, and this is what happens in $\Lpal$.  However, as we will see in the next two sections, DPDA and NPDA languages obey a kind of Pumping Lemma due to their nested nature, and we can use this to prove that some languages are beyond their ken.

%\pagebreak

\section{Context-Free Grammars}

The job of an automaton is to \emph{recognize} a language---to receive a word as input, and answer the yes-or-no question of whether it is in the language.  But we can also ask what kind of process we need to \emph{generate} a language.  In terms of human speech, recognition corresponds to listening to a sentence and deciding whether or not it is grammatical, while generation corresponds to making up, and speaking, our own grammatical sentences.

A \emph{context-free grammar} is a model considered by the Chomsky school of formal linguistics.  The idea is that sentences are recursively generated from internal mental symbols through a series of production rules.  For instance, if $S$, $N$, $V$, and $A$ correspond to sentences, nouns, verbs, and adjectives, we might have rules such as $S \to NVN$, $N \to AN$, and so on, and finally rules that replace these symbols with actual words.  We call these rules context-free because they can be applied regardless of the context of the symbol on the left-hand side, i.e., independent of the neighboring symbols.  Each sentence corresponds to a \emph{parse tree} as shown in Figure~\ref{fig:parse}, and \emph{parsing} the sentence allows us to understand what the speaker has in mind.

\begin{figure}
\begin{center}
\includegraphics[width=2in]{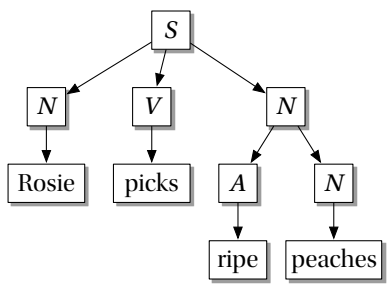}
\end{center}
\caption{The parse tree for a small, but delightful, English sentence.}
\label{fig:parse}
\end{figure}

Formally, a context-free grammar $G$ consists of a finite alphabet $V$ of \emph{variable} symbols, an initial symbol $S \in V$, a finite alphabet $T$ of \emph{terminal} symbols in which the final word must be written, and a finite set $R$ of \emph{production rules} that let us replace a single variable symbol with a string composed of variables and terminals:
\[ 
\mbox{$A \to s$ where $A \in V$ and $s \in (V \cup T)^*$} \, . 
\]
If $A \in V$ and $w \in T^*$, we write $A \leadsto w$ if there is a \emph{derivation}, or sequence of production rules, that generates $w$ from $A$.  
%We say that $w$ is \emph{grammatical} if $S \leadsto w$, in which case each production corresponds to a node of the parse tree, and 
We say that $G$ \emph{generates} the language $L(G) \subseteq T^*$, consisting of all terminal words $w$ such that $S \leadsto w$.

For instance, the following grammar has $V=\{S\}$, $T=\{ (, ), [, ] \}$, and generates the language $D_2$ with two types of brackets:
\begin{equation}
\label{eq:d2}
S \to (S)S, [S]S, \eps \, .
\end{equation}
This shorthand means that $R$ consists of three production rules.  We can replace $S$ with $(S)S$ or $[S]S$, or we can erase it by replacing it with the empty word $\eps$.  Using these rules we can produce, for instance, the word $([])[]$:
\[
S \to (S)S \to ([S]S)S \to ([S]S)[S]S \leadsto ([])[] \, ,
\]
where in the last $\leadsto$ we applied the rule $S \to \eps$ four times.

We say that a language $L$ is \emph{context-free} if it can be generated by some context-free grammar.  Clearly the bracket languages $D_k$ are context-free.  So is the language $\{ a^n b^n \mid n \ge 0 \}$, since we can generate it with
\[
S \to aSb, \eps \, . 
\]
Similarly, the language of palindromes $\Lpal$ is context-free, using
\[
S \to aSa, bSb, a, b, \eps \, . 
\]
Since all these examples have just a single variable, consider the language $L = \{ a^m b^n \mid m > n \}$.  We can generate it with the variables $V=\{S,U\}$ and the rules
\begin{align*}
S &\to aS, aU \\
U &\to aUb, \eps \, . 
\end{align*}

\begin{exercise}
Give a context-free grammar for $L_{a=b}$, the set of words with an equal number of $a$'s and $b$'s, appearing in any order.
\end{exercise}

\loesung{
There are many solutions, but here is the simplest:
\[
S \to aSbS, bSaS, \eps \, . 
\]
To see that this works, suppose that a word $w \in L_{a=b}$ begins with an $a$.  Then there must be a $b$ somewhere later in the word that matches that $a$ in the sense that $w=axby$ for some $x, y \in L_{a=b}$.  If we use the rule $S \to aSbS$, then we can use the two $S$'s to generate $x$ and $y$.  The case where $w$ begins with a $b$ is similar.
}

A few closure properties follow almost immediately from the definition of context-free grammars:
\begin{exercise}
Show that the context-free languages are closed under union, concatenation, reversal, and Kleene star.
\end{exercise}

\loesung{
Let $L_1$ and $L_2$ be context-free languages.  They each have a grammar, with initial symbols $S_1$ and $S_2$.  To generate $L_1 \cup L_2$, create a new grammar with initial symbol $S$ and the rules $S \to S_1, S_2$.  To generate $L_1 L_2$, use $S \to S_1 S_2$.  

To generate $L^R$, just modify the grammar for $L$ by reversing the strings on the right-hand side of each rule.

To generate the Kleene star $L^*$ of a language $L$ generated by a grammar with initial symbol $S$, just add the rules $S \to SS$ and $S \to \eps$.  By repeatedly applying $S \to SS$, we can generate $S^t$ for any $t$, and so the concatenation of any $t$ words in $L$.  The rule $S \to \eps$ lets us generate the empty word as well.
}

\noindent
But what about intersection and complement?

Any regular language is context-free, although the above examples show the converse is not the case.  Call a context-free language \emph{regular} if each of its production rules either lets us change one variable to another while leaving a single terminal behind it,
\begin{align*}
\mbox{$A \to tB$ where $t \in T$ and $B \in V$} \, ,
\end{align*}
or lets us end the production process by erasing the variable, $A \to \eps$.  Then we have the following theorem:
\begin{theorem}
A language is regular if and only if it can be generated by a regular grammar.
\end{theorem}

\begin{exercise}
Prove this theorem.  What do the variables correspond to?
%  Hint: the variables correspond to the states of an NFA.
\end{exercise}

It's sometimes convenient to reduce a context-free grammar to a particular form.  Consider these exercises:

\begin{exercise}
A grammar is in \emph{Chomsky normal form} if the right-hand side of every production rule is either a terminal or a pair of variables:
\[
\mbox{$A \to BC$ where $B, C \in V$, or $A \to t$ where $t \in T$} \, ,
\]
unless the only rule is $S \to \eps$.  Show that any context-free grammar can be converted to one in Chomsky normal form.
\end{exercise}

\begin{exercise}
A grammar is in \emph{Greibach normal form} if the right-hand side of every production rule consists of a terminal followed by a string of variables:
\[
\mbox{$A \to tw$ where $t \in T, w \in V^*$} \, ,
\]
unless the only rule is $S \to \eps$.  Show that any context-free grammar can be converted to one in Greibach normal form.
\end{exercise}

All right, enough pussyfooting around.  The context-free languages are precisely those that can be recognized by nondeterministic push-down automaton.  We will prove this in two pieces.

\begin{theorem}
\label{thm:npda-cfl}
Any language recognized by an NPDA is context-free.
\end{theorem}

\begin{proof}
It's convenient to demand that NPDA only accepts when the stack is empty.  If we are left with some residual symbols on the stack, we can always use a series of $\eps$-transitions to pop them off.  Alternately, the NPDA can simply choose, nondeterministically, not to push these symbols on the stack in the first place.

We will define a context-free grammar that generates accepting paths for the NPDA, and then converts these paths to input words that induce them.  This grammar has two kinds of variables.  Variables of the form 
\[
(s_1 \leadsto s_2)
\]
denote a computation path that starts in the state $s_1$ and ends in the state $s_2$, where $s_1, s_2 \in S$, and that leaves the stack in the same state it started in.  Those of the form 
\begin{equation}
\label{eq:var-steps}
(s_1 \stackrel[a]{\textrm{push}\ \gamma}{\longrightarrow} s_2) \, , \; 
(s_1 \stackrel[a]{\textrm{pop}\ \gamma}{\longrightarrow} s_2) \, , \;
(s_1 \stackrel[a]{}{\longrightarrow} s_2) 
\end{equation}
denote single steps of computation that read a symbol $a$, change the state from $s_1$ to $s_2$, and possibly perform an action on the stack, as allowed by the NPDA's transition function.  If you want to allow $\eps$-transitions, we can include variables with $\eps$ in place of $a$.

The initial symbol represents an accepting computation path, 
\[
S = (\sinit \leadsto \saccept) \, ,
\]
where for simplicity we assume that there is a single accepting state $\saccept$.  Our goal is to ``unpack'' this path into a series of single steps.  We can do this with the following kinds of rules:
\begin{align*}
(s_1 \leadsto s_2) &\to 
(s_1 \stackrel[a]{\textrm{push}\ \gamma}{\longrightarrow} s_3) \,(s_3 \leadsto s_4) \,(s_4 \stackrel[a]{\textrm{pop}\ \gamma}{\longrightarrow} s_2) \\
(s_1 \leadsto s_2) &\to
(s_1 \stackrel[a]{}{\longrightarrow} s_3) \,(s_3 \leadsto s_2) \, .
\end{align*}
Each of these rules refines the path.  The first rule starts by pushing $\gamma$ and ends by popping it, and the second one starts with a move that changes the state but not the stack.  In addition, a path that doesn't change the state can be replaced with no move at all:
\[
(s \leadsto s) \to \eps \, .
\]
Repeatedly applying these rules lets us generate any accepting computation path, i.e., any string of variables of the form~\eqref{eq:var-steps} that describe a series of state transitions, and pushes and pops on the stack, that are consistent with the NPDA's transition function.
 
Finally, our terminals are the input symbols $a \in A$.  We convert an accepting computation path to a string of terminals according to the rules
\[
(s_1 \stackrel[a]{\textrm{push}\ \gamma}{\longrightarrow} s_2) \, , \; 
(s_1 \stackrel[a]{\textrm{pop}\ \gamma}{\longrightarrow} s_2) \, , \;
(s_1 \stackrel[a]{}{\longrightarrow} s_2) \to a \, ,
\]
so that this grammar generates exactly the words accepted by the NPDA.
\end{proof}

Now for the converse:

\begin{theorem}
Any context-free language can be recognized by some NPDA.
\end{theorem}

\begin{proof}
For convenience we will let our NPDA make $\eps$-transitions.  We will also let it pop the top symbol and push any finite string in $\Gamma^*$ onto the stack, all in a single step.  We can simulate such automata with our previous ones by making a series of $\eps$-transitions, or by increasing the size of the stack alphabet (exercise for the reader).

\begin{figure}
\begin{center}
\includegraphics[width=5.5in]{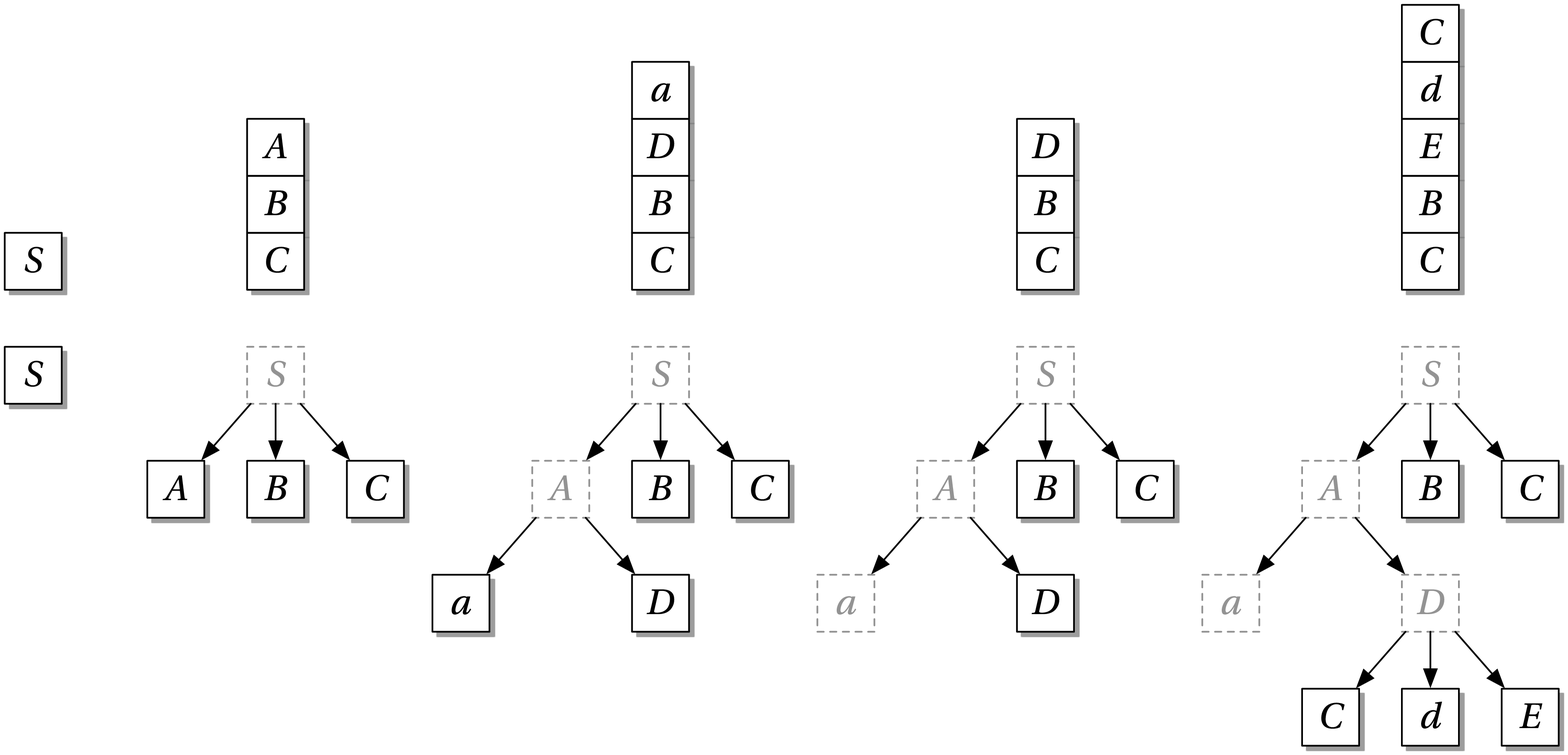}
\end{center}
\caption{An NPDA generating a parse tree.  The stack contains a string of currently active symbols, variables waiting to branch and terminals waiting to be matched with the input.  In the first step, the NPDA applies the rule $S \to ABC$.  In the second step, it applies $S \to aD$.  In the third step, it reads the first symbol of the input, and checks that it is an $a$.  In the third step, it applies the rule $D \to CdE$, and so on.  The gray symbols are no longer active, i.e., they are ``dead wood'' where the parse tree is no longer growing.}
\label{fig:grow}
\end{figure}

Imagine growing the parse tree, starting at the root.  At each point in time, there are ``buds,'' or variables, that can give birth to further branches, and ``leaves,'' or terminals, which grow no further.  The NPDA keeps track of these buds and leaves using the stack alphabet $\Gamma = V \cup T$, storing them so that the leftmost one is on top of the stack as in Figure~\ref{fig:grow}.  At each step, if the leftmost symbol is a terminal, the NPDA pops it, reads the next input symbol, and checks that they match.  But if the leftmost symbol is a variable, the NPDA performs an $\eps$-move that applies one of the production rules, popping the variable off the stack and replacing it with its children.  

Thus on any given step, the stack contains a string of variables waiting to branch further, and terminals waiting to matched to the input.  When the stack is empty, the entire parse tree has been generated and every terminal has been matched, and at that point the NPDA accepts.
\end{proof}

As with regular languages, each definition of the context-free languages makes certain things easier to prove.  For instance, it seems much easier to prove that they are closed under intersection with regular languages (Exercise~\ref{ex:closed-int-reg}) by thinking about NPDAs rather than about grammars.  

How hard is it to tell whether a given string $w$ can be generated by a given context-free grammar?  We will see in Problem~3.30 that we can do this in $\bigo(|w|^2 |V|)$ time using dynamic programming.  Do you see why?  If we try to solve the problem recursively by breaking it into subproblems, how many distinct subproblems will we encounter?

\begin{exercise}
Show that any context-free language over a one-symbol alphabet $\{a\}$ is regular.
\end{exercise}

%\pagebreak

\section{The Pumping Lemma for Context-Free Languages}

How can we prove that a language is not context-free?  As we pointed out above, we can't get anywhere by counting equivalence classes---a stack has enough freedom to store the entire word it has seen so far.  The question is what it can do with this information, given that it can only look deep in the stack by popping the other symbols first.  Intuitively, this means that context-free languages can only check for nested relationships, like those in the bracket languages, and not more complicated webs of dependence.

Happily, context-free grammars possess a kind of Pumping Lemma, much like the one we proved for regular languages in Section~\ref{sec:pumping}, and with a similar proof.  Here it is:
\begin{lemma}
\label{lem:cfl-pumping}
Suppose $L$ is a context-free language.  Then there is a constant $p$ such that any $w \in L$ with $|w| > p$ can be written as a concatenation of five strings, $w=uvxyz$, where
\begin{enumerate}
\item $|vy| > 0$, i.e., at least one of $v, y$ is nonempty, and
\item for all integers $t \ge 0$, $uv^txy^tz \in L$. 
\end{enumerate}
\end{lemma}

\begin{proof}
Consider a context-free grammar that generates $L$, and let $k$ be the largest number of symbols that any production rule turns a variable into.  Any parse tree of depth $d$ can produce a word of length at most $k^d$.  Thus if $|w| > k^{|V|}$, the parse tree must contain at least one path from the root of length $|V|+1$, containing $|V|+1$ variables starting with $S$.  By the pigeonhole principle, some variable $A$ must occur twice along this path.  

But as Figure~\ref{fig:cfl-pumping} shows, this means that 
\[
\mbox{$S \leadsto uAz$ \, , \; $A \leadsto vAy$ \, , and $A \leadsto x$} \, ,
\]
where $u, v, x, y, z$ are terminal words such that $uvxyz=w$.  Repeating the derivation $A \leadsto vAy$ $t$ times gives $S \leadsto uv^txy^tz$, showing that all these words are in $L$.

We can ensure that at least one of $v$ or $y$ is nonempty if no variable other than $S$ can produce the empty string.  Any context-free grammar can be converted to one with this property (exercise for the reader).  Finally, setting $p=k^{|V|}$ completes the proof.
\end{proof}

\begin{figure}
\begin{center}
\includegraphics[width=1.5in]{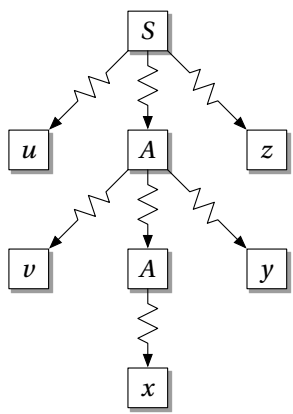}
\end{center}
\caption{The proof of the Pumping Lemma for context-free languages.  Any sufficiently deep parse tree has a path on which some variable $A$ appears twice, but this means that $A$ can generate $vAy$ for some terminal words $v, y$.  Repeating this process creates additional copies of $v$ and $y$ on the left and right, producing $uv^txy^tz$ for any $t \ge 0$.}
\label{fig:cfl-pumping}
\end{figure}

\begin{exercise}
\label{ex:vxy}
Strengthen Lemma~\ref{lem:cfl-pumping} to add the requirement that $|vxy| \le p$.  Hint: consider the last $|V|+1$ variables on the longest path in the parse tree.
\end{exercise}

Let's use the Pumping Lemma to prove that the language
\[
L_{a=b=c} = \left\{ w \in \{a,b,c\}^* \mid \#_a(w) = \#_b(w) = \#_c(w) \right\} 
\]
is not context-free.  (Note that it can be recognized by a two-counter automaton.)  We start by using the fact that the context-free languages are closed under intersection with regular languages.  Specifically, if $L_{a=b=c}$ is context-free, then so is its subset where the letters are sorted into blocks, 
\[
L' = L_{a=b=c} \cap a^* b^* c^* = \{ a^n b^n c^n \mid n \ge 0 \} \, . 
\]
If $L'$ is context-free, there is a constant $p$ such that any word of length greater than $p$ can be pumped as in Lemma~\ref{lem:cfl-pumping}.  But consider $w=a^p b^p c^p$.  If $w=uvxyz$ and $uv^2xy^2z$ is also in $L'$, then $v$ and $y$ must each lie inside one of the three blocks.  But there is no way to keep all three blocks balanced as we repeat $v$ and $y$.  We can keep the number of $a$s equal to the number of $b$s, and so on, but we can't keep all three numbers equal.  

Thus $L'$ is not context-free, and neither is $L_{a=b=c}$.  Context-free languages can only handle pairs of partners, not m\'enages \`a trois.  This example also lets us answer our questions about closure properties of context-free languages under Boolean operations:

\begin{exercise}
Show that the context-free languages are not closed under intersection.  Hint: show that $L'$ is the intersection of two context-free languages.   
\end{exercise}

\noindent
Since we have closure under union, if we had closure under complement then de Morgan's law would imply closure under intersection.  Thus the context-free languages are not closed under complement either.  For a direct proof of this, consider the following.

\begin{exercise}
Show that $\overline{L'}$ is context-free.
\end{exercise}

\noindent
On the other hand, the DPDA languages \emph{are} closed under complement, so\ldots

\begin{exercise}
Show that the DPDA languages are not closed under intersection or union.  
\end{exercise}

The next exercises illustrate that context-free languages can't handle relationships that cross each other, rather than nesting inside each other:

\begin{exercise}
\label{ex:abab}
Show that the language
\[
L = \{ a^m b^n a^m b^n \mid m, n \ge 0 \}
\]
is not context-free.  You will need the added requirement given by Exercise~\ref{ex:vxy} that $|vxy| \le p$.
\end{exercise}

\begin{exercise}
Recall that $\Lcopy$ is the language of words repeated twice in the same order, $\big\{ ww \mid w \in \{a,b\}^* \big\}$.  Show that $\Lcopy$ is not context-free.  Hint: take the intersection with a regular language to reduce this to Exercise~\ref{ex:abab}.
\end{exercise}

The copy language would be easy to recognize if we had an automaton with a queue, rather than a stack---like a line of people at the movies, symbols push at the back and pop at the front, in a first-in, first-out order.  Real-time queue automata, which take a single step for each symbol and must accept reject immediately at the end of the input, are incomparable with PDAs---things that are easy for a stack can be hard for a queue, and vice versa.  See Cherubini et al., \emph{Theoretical Computer Science} 85:171--203 (1991).  

On the other hand, the complement of the copy language is context-free:

\begin{exercise}
Give a context-free grammar for the complement of the copy language.  Note that words in $\overline{\Lcopy}$ of even length can be written in the form
\[
\mbox{$w=uv$ where $|u|=|v|$ and there is an $1 \le i \le |u|$ such that $u_i \ne v_i$} \, ,
\]
where $u_i$ denotes the $i$th symbol of $u$.  This is tricky: try to avoid thinking of $w$ as composed of two equal halves.
\end{exercise}

\section{Ambiguous and Unambiguous Grammars} 

In addition to the distinction between DPDAs and NPDAs, there is another nice distinction we can make between context-free grammars---namely, whether they generate each word with a unique parse tree.  In the linguistic sense, this is the question of whether every sentence can be clearly understood, unlike ambiguous phrases like 
``the bright little girls' school'' 
where it isn't clear which nouns are modified by which adjectives.

\begin{exercise}
Show that the grammar $S \to (S)S, [S]S, \eps$ for $D_2$ is unambiguous.
\end{exercise}

Note that a given language can have both ambiguous and unambiguous grammars.  For instance, the grammar $S \to S(S)S, \eps$ also generates $D_2$, but does so ambiguously---the word $()()$ has two distinct parse trees.  We call a context-free language \emph{unambiguous} if it possesses an unambiguous grammar.  There are also \emph{inherently ambiguous} context-free languages, for which no unambiguous grammar exists.  One example is
\[
\left\{ a^i b^j c^k \mid \mbox{$i=j$ or $j=k$} \right\} \, .
\]
We will not prove this, but any grammar for this language generates words with $i=j=k$ with two distinct parse trees, since $b^j$ could be ``meant'' to match either $a^j$ or $c^j$.

Converting a DPDA into a context-free grammar as we did for NPDAs in Theorem~\ref{thm:npda-cfl} shows that any DPDA language is unambiguous.  However, there are also unambiguous languages which cannot be recognized by a DPDA, such as 
\[
\left\{ a^i b^j \mid \mbox{$i > 0$ and $j=i$ or $j=2i$} \right\} \, . 
\]
Thus the context-free languages can be stratified into several subclasses, 
\[
\textrm{DPDA} \subset \textrm{unambiguous} \subset \textrm{NPDA} \, . 
\]

We can't resist pointing out that unambiguous grammars have lovely algebraic properties.  Given a language $L$, define its \emph{generating function} as 
\[
g(z) = \sum_{\ell=0}^\infty n_\ell z^\ell \, , 
\]
where $n_\ell$ is the number of words in $L$ of length $\ell$.  
\begin{exercise}
Using the generalization of this grammar to $k$ types of brackets, show that the generating function of $D_k$ obeys the quadratic equation
\begin{equation}
\label{eq:wilf}
g(z) = k z^2 g(z)^2 + 1 \, .
\end{equation}
Solve this equation, judiciously choosing between the two roots, and obtain
\[
g(z) = \frac{1-\sqrt{1-4kz^2}}{2 k z^2} \, , 
\]
and therefore, for all even $\ell$, 
\[
n_\ell = \frac{k^{\ell/2}}{\ell/2+1} {\ell \choose \ell/2} \, . 
\]
\end{exercise}

\begin{exercise}
Argue that the generating function of any unambiguous language obeys a system of polynomial equations with integer coefficients analogous to~\eqref{eq:wilf}.
\end{exercise}
\noindent 
As a result of this exercise, one way to prove that a language is inherently ambiguous is to show that its generating function is transcendental, i.e., not the root of a system of polynomial equations.  See Chomsky and Sch{\"u}tzenberger, ``The Algebraic Theory of Context-Free Languages,'' in \emph{Computer Programming and Formal Systems}, P. Braffort and D. Hirschberg (Eds.), 118--161 (1963).  You can learn more about generating functions in Wilf's lovely book \emph{generatingfunctionology}.

How accurate are context-free grammars as models of human languages?  Certainly real-life grammars, which let us place subordinate clauses, and other structures, in the middle of a sentence, have some stack-like structure.  Moreover, we feel that the claim that the stacks of English speakers, unlike those of German speakers whose depth is rumored to be five or so, have a depth of just one or two, while oft-repeated, is unjustified.  After all, the following phrase is perfectly understandable:
\begin{center}
This is the malt that the rat that the cat that the dog worried killed ate, that lay in the house that Jack built.
\end{center}
\noindent
On the other hand, it's not obvious to us that the human brain contains anything remotely resembling a stack.  It seems to be easier to remember strings of words in the same order we heard them, as opposed to reverse order, which suggests that we have something more like a queue.  There are some languages, like Swiss German, that have queue-like constructions of the form $abcabc$, and it isn't much harder for a human to recognize $\{a^n b^n c^n\}$ than it is to recognize $\{a^n b^n\}$.  Presumably we have a soup of mental symbols with varying urgency, creating and interacting with each other a wide variety of ways.

However, many computer languages are context-free, or nearly so.  In programming languages like C, Java, ML, and so on, we use brackets to open and close blocks of code, such as loops, if-then statements, and function definitions, and we use parentheses to build up mathematical expressions like $(x+y)*z$.  These languages are designed so that a compiler can parse them in linear time, reading through a program's source code and quickly converting it to machine code.

\section{Context-Sensitive Grammars}

In a context-free grammar, the left-hand side of each production rule is a single variable symbol.  \emph{Context-sensitive} grammars allow productions to depend on neighboring symbols, and thus replace one finite string with another.  However, we demand that the grammar is \emph{noncontracting}, in that productions never decrease the length of the string.  Thus they are of the form
\begin{equation}
\label{eq:noncon}
\mbox{$u \to v$ where $u,v \in (V \cup T)^*$ and $|u| \le |v|$} \, .
\end{equation}
If we like, we can demand that $u$ contain at least one variable.  To generate the empty word, we allow the rule $S \to \eps$ as long as $S$ doesn't appear in the right-hand side of any production rule.  We say that a language is context-sensitive if it can be generated by a context-sensitive grammar.

Here is a context-sensitive grammar for the copy language.  We have $V=\{S,C,A,B\}$, $T=\{a,b\}$, and the production rules
\begin{gather*}
S \to aSA, bSB, C \\
CA \to Ca, a \\
CB \to Cb , b \\
aA \to Aa \, , \quad aB \to Ba \, , \quad bA \to Ab \, , \quad bB \to Bb \, . 
\end{gather*}
These rules let us generate a palindrome-like string such as $abSBA$.  We then change $S$ to $C$ and use $C$ to convert upper-case variables $A, B$ to lower-case terminals $a, b$, moving the terminals to the right to reverse the order of the second half.  With the last of these conversions, we erase $C$, leaving us with a word in $\Lcopy$.  

For instance, here is the derivation of $abab$:
\begin{align*}
S &\to aSA \\
&\to abSBA \\
&\to abCBA \\
&\to abCbA \\
&\to abCAb \\
&\to abab \, . 
\end{align*}
Note that if we erase $C$ prematurely, we get stuck with variables $A, B$ that we can't remove.

\begin{exercise}
Give context-sensitive grammars for these languages.
\begin{enumerate}
\item $\left\{a^n b^n c^n \mid n \ge 0 \right\}$.
\item $\left\{ a^{2^n} \mid n \ge 0 \right\} = \{ a, aa, aaaa, aaaaaaaa, \ldots \}$.
\item $\{ 0, 1, 10, 11, 101, 1000, 1101, \ldots \}$, the Fibonacci numbers written in binary.  Hint: include the variables $\begin{pmatrix} 0 \\ 0 \end{pmatrix}$, $\begin{pmatrix} 0 \\ 1 \end{pmatrix}$, $\begin{pmatrix} 1 \\ 0 \end{pmatrix}$, and $\begin{pmatrix} 1 \\ 1 \end{pmatrix}$, 
\end{enumerate}
\end{exercise}

A stricter definition of a context-sensitive rule replaces a single variable $A$ with a nonempty string $s$, while allowing strings $u$ and $v$ on either side to provide context:
\begin{equation}
\label{eq:cs}
\mbox{$uAv \to usv$ where $A \in V$, $u,v,s \in (V \cup T)^*$, and $s \ne \eps$} \, .
\end{equation}
Any grammar of the form~\eqref{eq:noncon} can be converted to one of the form~\eqref{eq:cs}.  We won't ask you to prove this, but you may enjoy this exercise:

\begin{exercise}
By adding more variables, show how to convert the rule $AB \to BA$ into a chain of rules of the form~\eqref{eq:cs}.
\end{exercise}

\noindent
As you can see, context-sensitive grammars are very powerful.  We will see just how powerful in the next section.

\section{Turing Machines and Undecidability}

The machines we have seen so far read their input from left to right---they cannot revisit symbols of the input they have seen before.  In addition, the input is read-only---the machine can't modify the symbols of the input and use it as additional space for its computation.  In this section, we look at a kind of machine that can move left or right on the input, and read and write symbols anywhere it likes in an infinite workspace.  These machines are far more powerful that any we have seen before.  Indeed, they are as powerful as any machine can be.

In 1936, Alan Turing proposed a mathematical model of a computer---a human computer, carrying out an algorithm with pencil and paper.  This \emph{Turing machine} has a finite-state machine ``head'', and a ``tape'' where each square contains a symbol belonging to a finite alphabet $A$.  At each step, it observes the symbol at its current location on the tape.  Based on this symbol and its internal state, it then updates the symbol at its current location on the tape, updates its internal state, and moves one step left or right on the tape.  If $A$ denotes the tape alphabet and $S$ denotes the set of internal states, we can write the machine's behavior as a transition function
\[
\delta : A \times S \to A \times S \times \{ +1, -1 \} \, . 
\]
For instance, $\delta(a,s) = (a',s',+1)$ would mean that if the machine sees the symbol $a$ on the tape when it is in state $s$, then it changes the tape symbol to $a'$, changes its state to $s'$, and moves to the right.  We show this in Figure~\ref{fig:tm-one-step}.

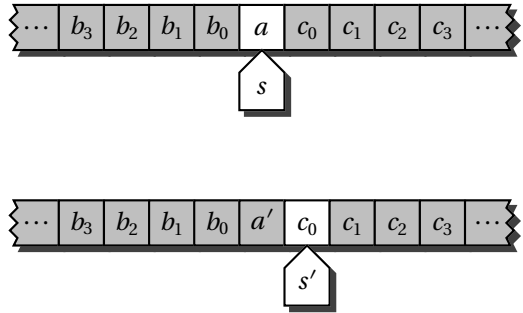
\begin{figure}
\centering
\begin{pspicture}(0,0)(7,4)
  %\psgrid
\psset{dimen=middle}
\psset{arrowsize=2pt 4}

\psset{fillstyle=solid,fillcolor=lightgray,shadow=true}
\pspolygon(0.2,3.4)(0.25,3.5)(0.15,3.6)(0.25,3.7)(0.15,3.8)(0.25,3.9)(0.2,4)(0.8,4)(0.8,3.4)
\rput(0.5,3.7){$\ldots$}
\psframe(0.8,3.4)(1.4,4)\rput(1.1,3.7){$b_3$}
\psframe(1.4,3.4)(2.0,4)\rput(1.7,3.7){$b_2$}
\psframe(2.0,3.4)(2.6,4)\rput(2.3,3.7){$b_1$}
\psframe(2.6,3.4)(3.2,4)\rput(2.9,3.7){$b_0$}
\psframe[fillcolor=white](3.2,3.4)(3.8,4)\rput(3.5,3.68){$a$}
\psframe(3.8,3.4)(4.4,4)\rput(4.1,3.66){$c_0$}
\psframe(4.4,3.4)(5.0,4)\rput(4.7,3.66){$c_1$}
\psframe(5.0,3.4)(5.6,4)\rput(5.3,3.66){$c_2$}
\psframe(5.6,3.4)(6.2,4)\rput(5.9,3.66){$c_3$}
\pspolygon(6.2,3.4)(6.8,3.4)(6.75,3.5)(6.85,3.6)(6.75,3.7)(6.85,3.8)(6.75,3.9)(6.8,4)(6.2,4)
\rput(6.5,3.7){$\ldots$}
  
\psset{fillcolor=white}
\pspolygon[dimen=outer](3.2,2.6)(3.2,3.1)(3.5,3.4)(3.8,3.1)(3.8,2.6)
\rput(3.5,2.9){$s$}

\psset{fillstyle=solid,fillcolor=lightgray,shadow=true}
\pspolygon(0.2,0.8)(0.25,0.9)(0.15,1.0)(0.25,1.1)(0.15,1.2)(0.25,1.3)(0.2,1.4)(0.8,1.4)(0.8,0.8)
\rput(0.5,1.1){$\ldots$}
\psframe(0.8,0.8)(1.4,1.4)\rput(1.1,1.1){$b_3$}
\psframe(1.4,0.8)(2.0,1.4)\rput(1.7,1.1){$b_2$}
\psframe(2.0,0.8)(2.6,1.4)\rput(2.3,1.1){$b_1$}
\psframe(2.6,0.8)(3.2,1.4)\rput(2.9,1.1){$b_0$}
\psframe(3.2,0.8)(3.8,1.4)\rput(3.5,1.15){$a'$}
\psframe[fillcolor=white](3.8,0.8)(4.4,1.4)\rput(4.1,1.06){$c_0$}
\psframe(4.4,0.8)(5.0,1.4)\rput(4.7,1.06){$c_1$}
\psframe(5.0,0.8)(5.6,1.4)\rput(5.3,1.06){$c_2$}
\psframe(5.6,0.8)(6.2,1.4)\rput(5.9,1.06){$c_3$}
\pspolygon(6.2,0.8)(6.8,0.8)(6.75,0.9)(6.85,1.0)(6.75,1.1)(6.85,1.2)(6.75,1.3)(6.8,1.4)(6.2,1.4)
\rput(6.5,1.1){$\ldots$}
  
\psset{fillcolor=white}

\pspolygon[dimen=outer](3.8,0)(3.8,0.5)(4.1,0.8)(4.4,0.5)(4.4,0)
\rput(4.1,0.35){$s'$}
  
\end{pspicture}
\caption{One step of a Turing machine.  It changes the tape symbol from $a$ to $a'$, changes its internal state from $s$ to $s'$, and moves one step to the right.}
\label{fig:tm-one-step}
\end{figure}

We start the machine by writing the input string $x$ on the tape, with a special blank symbol \blank\ written on the rest of the tape stretching left and right to infinity.  We place the head at the left end of the input, in some initial state $\sinit \in S$.  There is a special state $\saccept \in S$ where the machine halts its computation, and we say that it accepts the input $x$ if the computation path ends at $\saccept$.  There can be other halt states as well in which the machine rejects the input, but it can also reject by running forever.  

\begin{exercise}
Show that a Turing machine that can never move to the left can only recognize regular languages.
\end{exercise}

\begin{exercise}
Show that a Turing machine that never changes the symbols written on the tape can only recognize regular languages.  Hint: for each pair of adjacent tape squares, define the \emph{crossing sequence} as the history of the machine's moves left or right across the boundary between them, and its internal state at the time.  Show that, unless the machine falls into an infinite loop, only a finite number of different crossing sequences can occur.  Then describe how to check that these sequences can be woven together into a consistent computation path.  Allowing the machine to make excursions into the blank areas on either side of the input adds another wrinkle.
\end{exercise}

\begin{exercise}
Let a $k$-stack PDA be a finite-state automaton with access to $k$ stacks that reads its input left-to-right.  Show that if it is allowed to make $\eps$-transitions, taking additional computation steps after it has read the input, then a $k$-stack PDA can simulate a Turing machine whenever $k \ge 2$.  Show the same thing for a finite-state automaton with access to a single queue that can make $\eps$-transitions (unlike the real-time queue automata mentioned above).
\end{exercise}

Giving a Turing machine multiple heads, multiple tapes, a two-dimensional tape, a random-access memory, and so on doesn't increase its computational power.  Each of these new machines can be simulated by a Turing machine, although it might be considerably more efficient.  Indeed, Turing made a convincing philosophical case that these machines can carry out any procedure that could reasonably be called an algorithm.  They are equivalent to virtually every model of computation that has been proposed, including all the programming languages that we have today, and any physically realistic form of hardware that we can imagine.  Given enough time, they can even simulate quantum computers.

Turing then used a clever diagonal argument to show that there is no algorithm that can tell, in finite time, whether a given Turing machine (described in terms of its transition function) will accept a given input.  We say that the Halting Problem is \emph{undecidable}.  We can run the Turing machine as long as we like, but we cannot see into the infinite future and tell whether or not it will someday halt.  It is also undecidable whether a given Turing machine accepts the empty word---that is, whether it will halt and accept if its initial tape is entirely blank.

The fact that we can't predict the eventual behavior of a device capable of universal computation is perhaps not too surprising.  But the undecidability of the Halting Problem has consequences for simpler machines as well.  Just as we did for PDAs in Theorem~\ref{thm:npda-cfl}, let's write a Turing machine's computation path as a string of symbols of the form $(a,s,\pm 1)$, indicating at each step what state it moved to, what symbol it wrote on the tape, and which way it went.  Say that such a string is \emph{valid} if it is consistent with the machine's transition function and with the symbols on the tape, i.e., if symbols written on the tape are still there when the machine comes back to them, rather than surreptitiously changing their values.  Now consider the following:

\begin{exercise}
Show that the language of valid, accepting computation paths of a given Turing machine, given an initially blank tape, is
\begin{enumerate}
\item the intersection of two DPDA languages, and
\item the complement of a context-free language.
%\item the complement of a language recognized by a nondeterministic one-counter machine.
\end{enumerate}
Hint: break the tape into two halves.  
\end{exercise}

\begin{exercise}
Show that the following problems are undecidable, since if there were an algorithm for any of them, we could use it to tell whether a given Turing machine accepts the empty string:
\begin{enumerate}
\item given two context-free grammars $G_1$ and $G_2$, is $L(G_1) \cap L(G_2) = \emptyset$?
\item given two DPDAs $P_1$ and $P_2$, is $L(P_1) \cap L(P_2) = \emptyset$?
\item given two DPDAs $P_1$ and $P_2$, is $L(P_1) \subseteq L(P_2)$?
\item {\bf the universe problem:} given a context-free grammar or NPDA, does it generate or accept all possible words?
\item {\bf the equivalence problem:} given two context-free grammars or two NPDAs, do they generate or recognize the same language?
\end{enumerate}
\end{exercise}

\begin{exercise}
Strengthen the previous exercise to show that the language of valid, accepting paths is the complement of a language recognized by a nondeterministic one-counter machine.  Conclude that the universe and equivalence problems are undecidable for $1$-NCAs.
\end{exercise}

In particular, unlike DFAs, there is no algorithm that can take a context-free grammar, or an NPDA, or even a $1$-NCA, and find the minimal grammar or machine that is equivalent to it---since we can't even tell whether it is equivalent to a trivial grammar or machine that generates or accepts all words.  The equivalence problem (and therefore the universe problem) is decidable for DPDAs (and therefore for $1$-DCAs), although this was only shown quite recently; see S{\'e}nizergues, \emph{Theoretical Computer Science} 251:1--166 (2001).
%C. Stirling, Proc. ICALP 2002, 821--832.  
We will not prove this here, but it is also undecidable to tell whether a given context-free grammar is unambiguous.  The proof uses the undecidability of  the Post Correspondence Problem, see Section~7.6.3.

We can define nondeterministic Turing machines by letting the transition function be multiple-valued and accepting a word if an accepting computation path exists.  However, deterministic TMs can simulate nondeterministic ones.  Writing the simulation explicitly is tedious, but the idea is to do a breadth-first search of the tree of all possible computation paths, and accept if any of them accept.  Thus, as for lowly finite-state automata, nondeterminism adds nothing to the power of these machines---as long as we don't care how long they take.

Turing machines also have a grammatical description.  In an \emph{unrestricted} grammar, a production rule can convert any finite string in $(V \cup T)^*$ to any other.  
% (although we can require that at least one symbol on the left-hand side is a variable).  
These grammars can carry out anything from simple symbol-shuffling to applying the axioms of a formal system---the words they generate are the theorems of the system, the initial symbol and the production rules are the axioms, and the sequence of productions is the proof.  Unlike a context-sensitive grammar, rules can both increase and decrease the length of the string.  Thus the intermediate strings might be far longer than the string they generate, just as the proof of a theorem can involve intermediate statements far more complicated than the statement of the theorem itself.

Above, we spoke of a TM as starting with an input word written on the tape and accepting or rejecting it.  But a nondeterministic TM can also generate words, by starting with a blank tape, carrying out some sequence of transitions that write the word on the tape, and then halting.  If we treat the entire state of the TM as a string, with the head's state embedded in the tape, and extending this string whenever we move into a blank square we haven't visited before, we can transform its transition function to a grammar where each production carries out a step of computation.  Conversely, by writing and erasing symbols on the tape, a nondeterministic Turing machine can carry out the productions of an unrestricted grammar.

This shows that a language can be generated by an unrestricted grammar if and only if it can be generated by a nondeterministic TM.  But if a language can be generated by a TM $M$, it can also recognized by another TM that starts with $w$ on the tape, runs $M$ backwards (nondeterministically choosing the sequence of steps) and accepts if it reaches a blank initial tape.  Thus these three classes of language are identical:
\begin{enumerate}
\item languages generated by an unrestricted grammar, 
\item languages generated by some nondeterministic Turing machine,
\item languages recognized by some nondeterministic (or deterministic) Turing machine.
\end{enumerate}

If unrestricted grammars correspond to Turing machines, and context-free grammars correspond to push-down automata, what kind of machines do context-sensitive grammars correspond to?  A \emph{linear bounded automaton} or LBA is a nondeterministic Turing machine that is limited to the part of the tape on which its input is written---it may not move into the blanks on either side of the input.  However, it can have a tape alphabet $\Gamma$ larger than then input alphabet $A$, containing additional symbols it uses for computation.  

Since context-sensitive grammars never shorten the length of the string, they never need an intermediate string longer than the final string they generate.  By carrying out the production rules nondeterministically, an LBA can generate any word in a context-sensitive language with no more tape than the final word is written on.  Working backwards, an LBA can reduce an initial word to a blank tape, again without using more tape that it takes to store the word.  Thus a language is context-sensitive if and only if it can be generated, or recognized, by an LBA.  

What if the Turing machine has a little more room, say $C|w|$ squares for some constant $C$?  By enlarging the tape alphabet from $\Gamma$ to $\Gamma^C$, we can simulate such a machine with one with exactly $|w|$ squares, i.e., an LBA.  In modern terms, a context-sensitive language is one that can be recognized in $\bigo(n)$ space.  However, this is still an extremely powerful complexity class, including many problems we believe require exponential time.  Indeed, recognizing a given context-sensitive language is PSPACE-complete (see Chapter~8).

\begin{table}
\begin{center}
\begin{tabular}{l|l}
grammar & automaton \\ \hline
regular grammars & finite-state machines \\
context-free grammars & push-down automata \\
context-sensitive grammars & linear-bounded automata \\
unrestricted grammars & Turing machines
\end{tabular}
\end{center}
\caption{The Chomsky hierarchy.}
\label{tab:chomsky}
\end{table}

We have seen four classes of grammars that correspond to four classes of machines as shown in Table~\ref{tab:chomsky}.  In 1956, Noam Chomsky proposed these classes as a hierarchy of models for human language.  Since then, many other types of grammars and automata have been proposed---subclasses of the context-free grammars that are particularly easy to parse, and others with varying degrees of context-sensitivity.  Some of these intermediate classes even show up in chaotic dynamical systems: see Moore and Lakdalawa, \emph{Physica D} 135:24--40 (2000).

However, as our comments about stacks and queues above illustrate, computational complexity is not just a hierarchy, but a heterarchy.  Complexity classes can be incomparable, with each containing problems outside the other.  Just as stacks and queues have different strengths and weaknesses, different computational resources act in different ways---algorithms that are efficient in terms of time can be inefficient in terms of memory, and vice versa.

The Chomsky hierarchy is just a tiny piece of an infinite complexity hierarchy---and now it's time to read the rest of the book!

\end{document}